\documentclass[10pt,twocolumn,aps,pra,superscriptaddress]{revtex4-1}

\usepackage{graphicx}
\usepackage{color}
\usepackage{xcolor}
\usepackage{amssymb}
\usepackage{amsmath}
\usepackage{bm}
\usepackage[american]{babel}
\usepackage{hyperref}
\usepackage{braket}
\usepackage[T1]{fontenc} 
\usepackage{comment}
\usepackage{textcomp}
\usepackage[ruled,vlined]{algorithm2e}
\usepackage{amsthm}
\usepackage{txfonts}

\DeclareGraphicsExtensions{.pdf,.png,.jpg}
\hypersetup{colorlinks=true,urlcolor=teal,breaklinks=true,citecolor=teal,linkcolor=teal,pdfstartview=FitH,pdfpagemode=UseNone}

\newtheorem{theorem}{Theorem}
\newtheorem{lemma}{Lemma}

\theoremstyle{remark}
\newtheorem{remark}{Remark}

\newcommand{\kt}[1]{{\color{blue}{KT: #1}}}

\newcommand{\nz}{N_{\textrm{z}}}
\newcommand{\nzb}{N_{\textrm{z,bulk}}}
\newcommand{\nzs}{N_{\textrm{z,side}}}
\newcommand{\nq}{L_{\textrm{q}}}

\newcommand{\nr}{N_{\textrm{rounds}}}

\begin{document}

\title{Parity-unfolded distillation architecture for noise-biased platforms}


\author{Konstantin Tiurev}
\affiliation{Parity Quantum Computing Germany GmbH, 20095 Hamburg, Germany}

\author{Christoph Fleckenstein}
\affiliation{Parity Quantum Computing GmbH, A-6020 Innsbruck, Austria}

\author{Christophe Goeller}
\affiliation{Parity Quantum Computing Germany GmbH, 20095 Hamburg, Germany}
\affiliation{Parity Quantum Computing France SAS, 75016 Paris, France}

\author{Paul Schnabl}
\affiliation{Institute for Theoretical Physics, University of Innsbruck, A-6020 Innsbruck, Austria}

\author{Matthias Traube}
\affiliation{Parity Quantum Computing Germany GmbH, 20095 Hamburg, Germany}

\author{Nitica Sakharwade}
\affiliation{Parity Quantum Computing GmbH, A-6020 Innsbruck, Austria}

\author{Anette Messinger}
\affiliation{Parity Quantum Computing GmbH, A-6020 Innsbruck, Austria}

\author{Josua Unger}
\affiliation{Parity Quantum Computing GmbH, A-6020 Innsbruck, Austria}

\author{Wolfgang Lechner}
\affiliation{Parity Quantum Computing Germany GmbH, 20095 Hamburg, Germany}
\affiliation{Institute for Theoretical Physics, University of Innsbruck, A-6020 Innsbruck, Austria}
\affiliation{Parity Quantum Computing France SAS, 75016 Paris, France}
\affiliation{Parity Quantum Computing GmbH, A-6020 Innsbruck, Austria}

\date{\today}

\begin{abstract}
We introduce the parity-unfolded architecture, a fault-tolerant quantum computing scheme that relies on direct preparation and teleportation of small-angle rotations $ Z^{1/2^{k}}$ rather than approximating them with the conventional (Clifford + $T$) gate set. The architecture is enabled by efficient distillation of gates from an arbitrary level of the Clifford hierarchy, which we refer to as parity unfolding. With it, a state $\ket{Z_k} = Z^{1/2^{k}}\ket{+}$ can be prepared fault-tolerantly using $2^{k+3} + \mathcal{O}(2^{k/2})$ biased-noise qubits on a planar chip with nearest-neighbour connectivity. For algorithms requiring native $Z^{1/2^{k}}$ gates, such as the Quantum Fourier Transform and phase estimation, the proposed scheme allows to reduce resource overheads for up to $k=7$, i.e., up to $T^{1/32}$. Furthermore, when used for the synthesis of arbitrary small-angle rotations, parity-unfolded distillation of ($T$ + $\sqrt{T}$) reduces the minimum achievable logical error rate by 43\% while cutting the resource requirements by 26\%, when compared to unfolded distillation of only the $T$ gate. 
\end{abstract}

\maketitle


\section{Introduction}

The Eastin--Knill theorem and the related no-go restrictions demonstrate that no quantum error-correcting code can implement a full universal gate set solely via transversal operations, in which physical gates act independently on corresponding physical qubits of a code block~\cite{PhysRevLett.102.110502}. This limitation is particularly relevant for fault-tolerant architectures based on 2D topological stabilizer codes, such as the surface or color codes, because in those geometries any gate implemented by a constant-depth local circuit is constrained to lie within the second level of the Clifford hierarchy, i.e., the Clifford group~\cite{PhysRevLett.110.170503}. 

To reach full universality, the Clifford set is typically augmented with at least one non-Clifford gate. One of the most common approaches involves a $T$ gate, which can be achieved by distillation~\cite{PhysRevA.71.022316,Litinski_2019,PhysRevA.86.052329} and injection~\cite{PhysRevA.62.052316} of the magic $\ket{T}$ state, or using techniques such as code switching~\cite{PRXQuantum.5.020345,PRXQuantum.2.020341,PhysRevLett.113.080501,PhysRevResearch.7.023080} and gauge fixing~\cite{PhysRevA.95.032338,Bombín_2015,Jones_2016}. The (Clifford + $T$) gate set is universal and, by virtue of the Solovay--Kitaev theorem~\cite{9781107002173}, allows approximation of any target unitary to arbitrary precision, with only polylogarithmic overhead in the error. 

Preparing high-fidelity magic $\ket{T}$ states remains one of the most resource-intensive tasks in fault-tolerant quantum computing. Conventional magic state distillation takes many noisy magic states as input and produces fewer output states with higher fidelity. Because distillation circuits themselves are not noise-free, this process typically must be carried out at the logical level, leading to substantial space and time overheads, despite significant improvements since the initial proposal~\cite{Litinski_2019}. More recent  state cultivation schemes rely on the measurement of a Clifford operator and allow to minimize overheads further, however, they still require a large time overhead due to a small post-selection probability~\cite{Chamberland2020,Gidney_2019,gidney2024magicstatecultivationgrowing,xu2026distillingmagicstatesbicycle}. Under certain conditions, the unfolded distillation protocol~\cite{ruiz2025unfoldeddistillationlowcostmagic} offers a much more hardware-efficient way to distill the $\ket{T}$ state on the physical-qubit level. The  protocol exploits the asymmetry in error types in biased-noise qubits, such as cat qubits~\cite{Mirrahimi_2014,PhysRevX.9.041053,10.1038/s41586-024-07294-3}, to drastically reduce resource overheads. The scheme requires only a few tens of physical qubits to produce a high-fidelity $\ket{T}$ state, representing an order-of-magnitude reduction in qubit count and circuit volume relative to leading standard approaches~\cite{Litinski_2019,Gidney_2019,gidney2024magicstatecultivationgrowing}. This dramatic resource saving is promising as it could significantly lower one of the major bottlenecks in constructing scalable, universal fault-tolerant quantum computers. 

A universal gate set can be further supplemented with gates from higher levels of the Clifford hierarchy~\cite{hu2021climbingdiagonalcliffordhierarchy,PhysRevA.95.012329,xu2026controlledjumpcliffordhierarchy}, e.g., those of the form $Z^{1/2^{k}}$ with $k>2$. Although the overhead required to distill and teleport high-fidelity copies of higher-level states grows, the subsequent compilation overhead to synthesise arbitrary rotations shrinks. Overall,  the extended gate sets reduce the total space-time cost required to synthesise arbitrary angle $U(3)$ rotation gates compared to the standard Clifford $+ T$ gate set. Furthermore, $Z^{1/2^{k}}$ rotations arise in many quantum algorithms, for example those that make use of the quantum Fourier transform. Direct access to native $Z^{1/2^{k}}$ gates might pave a way to even more resource-efficient execution for these algorithms. Universal quantum computing architectures based on extended gate sets with $k>2$ have been considered in literature before, see, e.g., Refs.~\cite{landahl2013complexinstructionsetcomputing,PhysRevA.91.042315,Campbell_2016,mooney2021cost}. However, due to significant resource overheads, demanding hardware requirements, and increased error rates of direct distillation of gates from higher levels of the Clifford hierarchy, such schemes have seen limited practical adoption to date.

\begin{figure*}[t]
\includegraphics[width=0.95\textwidth]{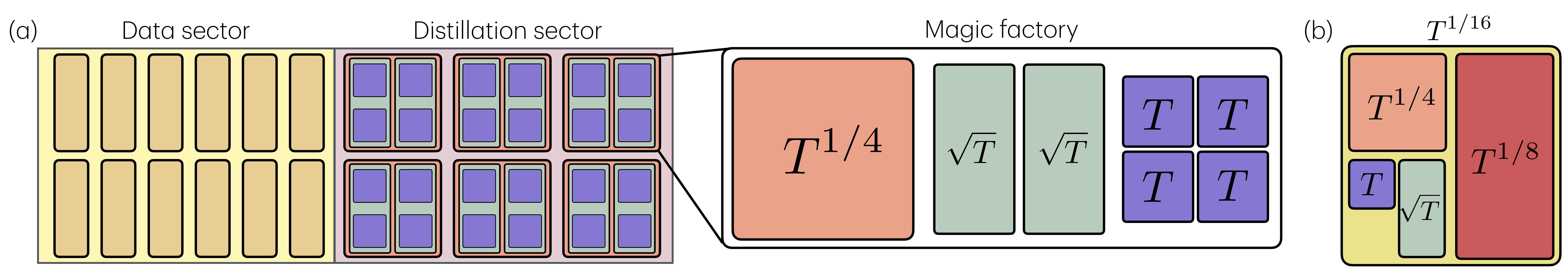}
\caption{\label{fig:scheme} High-level overview of the parity-unfolded architecture. (a)~A planar chip is divided into data and distillation sectors. Logical qubits of the data sectors hold information of the quantum algorithm. Magic factories are patches of qubits in the distillation sector. The inset illustrates magic factories containing 120 physical qubits and capable of producing either (i)~four $\ket{T}$, (ii)~two $\ket{\sqrt{T}}$, or (iii)~one $\ket{T^{1/4}}$ states in parallel. (b)~Magic factory for distillation and deterministic teleportation of $Z_{k\leq6}$ gates. In the first cycle, a gate $T^{1/16}$ is distilled on a magic factory~(yellow patch) and teleported onto a target logical qubit of the data sector. With a probability of 50\%, the teleportation requires a corrective $T^{1/8}$ gate, which can be distilled on a red patch and subsequently teleported, also with a success probability of 50\%. In the worst case, the entire sequence of gates down to the $T$ gate might be required. However, since a patch required for the distillation of a gate $Z_k$ only occupies a half of the patch required for the distillation of a gate $Z_{k+1}$, all potentially required corrective gates can be prepared in parallel within a second cycle of the protocol. Hence, a two-cycle protocol guarantees the availability of all corrective gates, making it deterministic.}
\end{figure*}

In this work, we devise hardware-friendly protocols for distilling gates from an arbitrary level of the Clifford hierarchy. The protocols are enabled by noise-biased characteristics of constituent physical qubits and realized by mapping face stabilizers of the Quantum Reed--Muller~(QRM)~\cite{gong2024computationquantumreedmullercodes,Quan_2018} codes onto classical planar parity codes, a procedure which we refer to as parity unfolding. We provide a constructive scheme for realizing such magic factories in planar chips with nearest-neighbour connectivity. The proposed scheme fault-tolerantly prepares any desired $Z^{1/2^{k}}$ rotation from a magic factory containing approximately $\mathcal{O}(2^{k+3})$ biased-noise physical qubits. This approach offers a new and practical design principle of fault-tolerant architectures, where gates from arbitrary high levels of the Clifford hierarchy are prepared directly in a parity-unfolded magic factory rather than approximated using the standard (Clifford + $T$) set. Fig.~\ref{fig:scheme} illustrates the design of the parity-unfolded architecture. We numerically simulate the full execution pipeline necessary to distill high fidelity magic states from different levels of the Clifford hierarchy. We use the distilled fault-tolerant gates to synthesise arbitrary rotations and demonstrate a clear advantage in logical error rate and space-time cost in experimentally relevant regimes. In particular, we find that using the extended parity-unfolded (Clifford + $T$ +$\sqrt{T}$) set 
for approximating an arbitrary single-qubit unitary yields a 38\% saving of resource cost at the same logical error rate when compared to the conventional (Clifford +$T$) set. Alternatively, the (Clifford + $T$ +$\sqrt{T}$) set allows for simultaneous reduction of resource cost and logical error rate by 26\% and 49\%, respectively. Furthermore, we show that algorithms that require only native non-Clifford gates of the form $Z^{1/2^{k}}$ benefit from our scheme even more.

The paper is organized as follows. In Section~\ref{sec:overview}, we provide a brief overview of the parity-unfolded architecture to motivate the development of magic states for higher levels of the Clifford hierarchy. The section also outlines the main results derived throughout the paper. Sections~\ref{sec:qrm}--\ref{sec:unfolding} are devoted to the design of unfolded magic factories, a central building block of our architecture. As such, Sec.~\ref{sec:qrm} provides background on the construction of Calderbank, Shor and Steane~(CSS) and QRM codes. In Sec.~\ref{sec:labels}, we revisit the most well-known 15-to-1 distillation scheme and its unfolding to a planar layout. In Sec.~\ref{sec:unfolding}, we introduce a constructive procedure to compose unfolded layouts for arbitrary levels of the Clifford hierarchy.  
The performance of parity-unfolded magic factories is analyzed analytically and simulated numerically in Sec.~\ref{sec:error-analysis}. In Sec.~\ref{sec:resource}, we estimate resource overheads of executing fault-tolerant rotations using the parity-unfolded architecture with an extended gate set and benchmark them against the conventional scheme based on $T$ state distillation. In the same section, we generalize the existing compilation algorithms to take into account gates from extended sets.Section~\ref{sec:conclusion} concludes the paper.

\section{Parity-unfolded architecture overview}\label{sec:overview}

Figure~\ref{fig:scheme} outlines the parity-unfolded architecture. As per standard fault-tolerant design, a planar chip has dedicated data and distillation sectors. The former holds qubits participating in logical computation. In a setting with very strongly noise-biased qubits, these logicals can be encoded in classical low-density parity check~(LDPC) codes, such as the repetition, Lechner--Hauke--Zoller~(LHZ)~\cite{PhysRevLett.129.180503,tiurev2025optimaldecodererrorcorrecting}, or LDPC-cat~\cite{ldpc-cat-code} codes. When bias is more moderate, noise-tailored quantum codes can be used~\cite{PhysRevLett.133.110601,XZZX,Lee2021surfacecode,leroux2025romanescocodesbiastailoredqldpc}. 

The distillation sector contains magic factories for distilling magic states of the form $\ket{Z_k} = Z^{1/2^{k}}\ket{+}$. Each factory takes a collection of noisy $\ket{Z_k}$ states~(or, in our case, requires a number of noisy $Z_k$ gates) and produces a single but less noisy state. High-quality magic states are subsequently consumed in a quantum teleportation protocol that applies a non-Clifford $Z^{1/2^{k}}$ rotation to one or a collection of logical qubits in the data sector. Unlike the conventional $T$ gate distilleries, each magic factory of the parity-unfolded architecture produces non-Clifford gates from all levels of the Clifford hierarchy up to level $(k+1)$. As an example, each factory of the distillation sector illustrated in Fig.~\ref{fig:scheme} gives an access to ($T$ + $\sqrt{T}$ + $T^{1/4}$) set of fault-tolerant non-Clifford gates. 

Our main finding is a resource-efficient and hardware-friendly design of magic factories for producing gates from arbitrary high levels of the Clifford hierarchy. The working principle of our magic factories relies on transversality of non-Clifford gates in quantum Reed--Muller~(QRM) codes and intrinsic noise bias of constituent qubits. We provide a step-by-step explicit construction of such factories in Secs.~\ref{sec:qrm}--\ref{sec:unfolding}. All gates up to and including $Z_k$ can be produced from a magic factory that contains approximately $2^{k+3}$ physical qubits~(including ancilla qubits that will be used for error correction). In particular, each such magic factory can distill $2^{k-k'}$ states $\ket{Z_{k'}}$ in parallel for any $k' \leq k$. As an example, magic factories of the distillation sector in Fig.~\ref{fig:scheme} contain each 120 physical qubits~(64 data and 56 ancillas) and can distill either of the following in parallel: (i)~four $\ket{T}$, (ii)~two $\ket{\sqrt{T}}$, (iii)~one $\ket{T^{1/4}}$ state, or any combinations of them with the combined qubit footprint not exceeding 120 qubits. As we show in Sec.~\ref{sec:error-analysis}, distilling either of the three gate sets requires the same number of error correction cycles on the same magic factory, yielding identical space-time cost. Hence, one can distill a state from the $(k+1)$st level of the Clifford hierarchy at a cost of two gates from the $k$th level.  In general, since distilling the gate $Z_{k+1}$ requires approximately twice the number of qubits required for the $Z_{k}$ gate, the sets of gate that can be distilled in parallel are determined by a number of patches that can be fit into the magic factory, as illustrated in Fig.~\ref{fig:scheme}~(b). The cost of distilling higher-level state is therefore rapidly grows. However, as we show in Sec.~\ref{sec:resource}, the additional overhead of distilling higher-level non-Clifford states is offset by the reduced number of states required for several important quantum computing subroutines.

To execute a quantum gate $Z_k$ on a logical data qubit, the gate must be teleported from the distillation sector. With a probability of 50\%, a gate $Z_k^{\dagger}$ will be applied instead, and needs to be corrected by applying $Z_k^2 = Z_{k-1}$. For the (Clifford + $T$) gate set, the only required correction operator is $T^2=S$, i.e., a Clifford gate that is considered cost-free. On the other hand, teleporting the distilled gate $Z_{k>2}$ might require gates from a sequence $Z_{k-1}, Z_{k-2},..Z_1$. The probabilistic nature of gate teleportation hence might introduce additional overhead in this case. However, it is easy to see that all such lower-level gates are simultaneously available with an additional overhead upper bounded by a factor of 2. Indeed, as we illustrate in Fig.~\ref{fig:scheme}~(b), all potentially required corrective gates fit into a magic factory used to distill $Z_k$ and can be prepared in the second distillation cycle. Hence, a maximum of two distillation cycles are required to apply a rotation $Z_k$ deterministically.

\section{Background:\\ transversality, Clifford hierarchy, and QRM codes}\label{sec:qrm}

We start by introducing the notations used throughout. We use $Z_k$ to denote Pauli-$Z$ rotations of the form
\begin{equation}\label{eq:rotation}
    Z_k :=
    \begin{pmatrix} 
    1 & 0 \\ 0 & e^{i\pi/2^k} 
    \end{pmatrix} 
    =
    e^{i\pi/2^{k+1}} 
    R_Z
    \!\left(\frac{\pi}{2^{k}}\right)
\end{equation}
with $k \geq 0$. As such, $k=0$ corresponds to a Pauli-$Z$ gate, $k=1$ to a Clifford $S$ gate, and $k=2$ to a non-Clifford $T$ gate. All distillation schemes considered in this paper rely on transversality---and, hence, inherent fault-tolerance---of non-Clifford gates in error-correcting QRM codes. In particular, it is known that a shortened $\overline{QRM}(1,k+2)$ code supports a transversal implementation of the gate $Z_k$. That is, applying the rotation $Z_k$ to each physical qubit of the $\overline{QRM}(1,k+2)$ code realizes the corresponding logical rotation on the encoded qubit. The most well-known example that supports a transversal ${T}$ gate is $\overline{QRM}(1,4)$ shown schematically in Fig.~\ref{fig:distillation}~(b).

A quantum gate corresponding to the rotation $Z_k$ is said to belong to the $(k+1)$st level of the Clifford hierarchy. As such, Pauli $Z_0 = Z$, Clifford $Z_1 = S$, and non-Clifford $Z_2=T$ gates belong to, respectively, the first, second, and third levels of the hierarchy. Since a gate $Z_k$ belongs to the $(k+1)$st level of the Clifford hierarchy and is transversal in $\overline{QRM}(1,k+2)$ a scheme for distilling a gate $Z_k$, a scheme for distilling a gate from the $(k+1)$st level of the hierarchy, and a scheme based on $\overline{QRM}(1,k+2)$ code are considered synonyms throughout the text, i.e.,
\[
Z_k
\; \leftrightarrow \;
\text{Clifford hierarchy level } (k+1)
\; \leftrightarrow \;
\overline{QRM}(1,k+2).
\]
To avoid confusion, we will use $k$ to denote the gate $Z_k$ in Eq.~\eqref{eq:rotation}, and $m$ to denote the shortened QRM code with parameters $\overline{QRM}(1,m)$. With this, $k$-order rotations are transversal in $\overline{QRM}(1,m)$ with $m=k+2$, which is always implied when $k$ and $m$ are used in the same equation. To avoid using the third set of labels, we refrain from classifying gates with levels of the Clifford hierarchy in the rest of the paper. 

The remainder of the section provides a brief overview of QRM codes that support various fault-tolerant non-Clifford gates. A more complete introduction to the theory of QRM codes can be found in Refs.~\cite{landahl2013complexinstructionsetcomputing}. 

The classical Reed–Muller code $RM(r,m)$ is defined as the set of all Boolean polynomials in $m$ binary variables of total degree at most $r$, evaluated on all $2^m$ binary inputs. A $RM(r,m$) code has block length $2^m$, distance $2^{m-r}$, and dimension 
\begin{equation} \label{eq:dimension}
    k = \sum_{i=0}^{r}\binom{m}{i}.    
\end{equation}
In standard coding theory notation, we say that $RM(r,m$) is a classical code with parameters
\begin{equation}
    [n,k,d]
    =
    [
    2^m, \;
    \sum_{i=0}^{r}\binom{m}{i}, \;
    2^{m-r}
    ],
\end{equation}
where $n$ is the number of physical (qu)bits, $k$ the number of logical (qu)bits and $d$ the distance of the code.
Importantly, the dual of a $RM(r,m$) is 
\begin{equation} \label{eq:duality}
{RM}^{\perp}(r,m) 
= 
{RM}(m-r-1,m).
\end{equation}

The CSS family is a subset of QEC codes with each stabilizer being made up either exclusively of $X$-type Pauli operators, or exclusively of $Z$-type Pauli operators. QRM codes, denoted as $QRM(r,m)$, arise naturally as the CSS codes whose defining $X$ and $Z$ parity check matrices are the generator matrices for $RM(r, m$) and its dual $RM(m-r-1,m$), respectively. Owing to duality of Eq.~\eqref{eq:duality}, this construction ensures that all $X$-type and $Z$-type stabilizers commute, which is necessary for a valid quantum code. 

In this paper, we will consider the family of first-order QRM codes defined by $r=1$. The first-order $QRM(1,m)$ code has parameters $[[2^m,0]]$, that is, it encodes no logical qubits. An operator with weight $2^{m}$~(which is equivalent to one of the  weight-$2^{m-1}$ stabilizers) can be promoted to logical operators, yielding the code with parameters
\begin{equation}\label{eq:full-qrm-parameters}
    [[n,k,d_x,d_z]]
    =
    [[2^m,1,d_x=2^{m-1},4]]
\end{equation}
It is sometimes more convenient to work with the shortened QRM code, which can be obtained from the full construction by puncturing one qubit and removing one stabilizer, which introduces a single logical degree of freedom. A shortened code $\overline{QRM}(1, m)$ is a CSS code formed by $2^m-2m-2$ type-$Z$ stabilizers~(which we will also call face stabilizers) and $m$ type-$X$~(which we will also call volume stabilizers) stabilizers. It has parameters 
\begin{equation}\label{eq:qrm-parameters}
    [[n,k,d_x,d_z]] = [[2^m-1,1,2^{m-1}-1,3]].
\end{equation}
The full code can be constructed from the shortened one by adding back one physical qubit and measuring the removed stabilizer. Hence one can see a full QRM code as a shortened code maximally entangled with the punctured qubit by the extra stabilizer. 
We will work with the shortened $\overline{QRM}(1,m)$ code throughout the paper, but occasionally use the full $QRM(1,m)$ when it is convenient for derivations. 

\section{Quantum state distillation using QRM codes}\label{sec:labels}

Shortened $\overline{QRM}(1,m)$ code can be used for distillation of a $Z_{m-2}$ gate that it supports transversally. A typical scheme is shown in Fig.~\ref{fig:distillation}. A QRM-encoded logical qubit is initially entangled with an external qubit we call target, which will hold the state $\ket{Z_{m-2}}$ at the end of the protocol. Encoding to the QRM code can be done in two ways. The first option is to use encoding circuits, such as the ones shown in Ref.~\cite{footnote-circuits} for $\overline{QRM}(1,4)$ and $\overline{QRM}(1,5)$. These circuits guarantee eigenvalues of certain multi-qubit Pauli operators to be $(+1)$, that is, set values of the code stabilizers to $(+1)$ without directly measuring them. Circuit-based encoding assumes noiseless qubits and multi-qubit gates during the encoding circuit. The scheme is hence typically executed on a logical level, where each qubit of the circuit in Fig.~\ref{fig:distillation} is itself encoded in, e.g., the surface code, and multi-qubit gates are executed in a fault-tolerant manner using protocols such as lattice surgery. 

Subsequently, a non-Clifford gate $Z^{\dagger}_{m-2}$ is applied transversally to each qubit of the code, resulting in a logical $\ket{\tilde{Z}^{\dagger}_{m-2}}$ state of the QRM-encoded qubit, where we used tilde to denote a state of a logical qubit. Upon measurement of all code qubits in the $X$ basis, the target qubit is left in the $\ket{\tilde Z_{m-2}}$ state. Unlike elements of the encoding circuit, $Z^{\dagger}_{m-2}$ gates are noisy and introduce errors even when executed on the logical level. To detect errors occurring due to noisy $Z^{\dagger}_{m-2}$ gates, one can reconstruct $X$-type stabilizers from the qubit measurements; a measured $(-1)$ eigenstate in at least one of these stabilizers indicates an error, and the corresponding distillation run shall be discarded. In $\overline{QRM}(1,m)$ codes, $X$ stabilizers can detect up to two errors, yielding the logical error rate $O(p_{\textrm{eff}}^3)$ with $p_{\textrm{eff}}$ being the effective probability of $Z$ or $Y$ error accumulated on each physical qubit prior to the final $X$ measurement. Since any configuration with one or two errors will be discarded, the probability to accept the outcome from the distillation protocol in the leading order reads $P_{\textrm{success}} = 1 - N_\text{qubits} p$, where $N_\text{qubits} = 2^m - 1$ is the number of physical qubits in the code. Despite tremendous progress in reducing the space-time cost, logical-level distillation remains a very resource-demanding sub-routine in fault-tolerant quantum computing. 

The second encoding scheme relies on direct measurement of the code stabilizers. As noted, the encoding scheme above is only applicable in the case where errors occur exclusively due to the noisy $Z^{\dagger}_{m-2}$ gates, while the encoding circuit is perfect. If, on the other hand, the encoding circuit is executed on the physical level, all components, including initialization and Clifford gates, are inevitably prone to noise. Such noisy encoding can transform less than three detectable errors into an undetectable error pattern, as well as introduce undetectable errors. 
To detect and correct errors during encoding, one can directly measure stabilizers of the shortened QRM code, as in standard quantum error correction. A practical complication in measuring stabilizers comes from the fact that a $\overline{QRM}(1,m)$ code is known to be local in $(m-1)$-dimensional space. Since the lowest-level non-Clifford $T$ gate is transversal in the $\overline{QRM}(1,4)$ code, it requires measurement of weight-8 $X$-stabilizers that can not be locally embedded in two-dimensional architectures, as examplified in Fig.~\ref{fig:distillation}~(b). 

For hardware platforms where qubits are subject to biased noise, unfolded distillation~\cite{ruiz2025unfoldeddistillationlowcostmagic} offers an elegant solution to the stabilizers non-locality problem. Since noise is strongly dominated by only one Pauli component, measuring stabilizers of that type is not necessary during encoding. Indeed, $X$-stabilizers always commute with errors that are fully biased towards Pauli-$X$. Only measurements of $Z$ stabilizers are necessary to detect errors taking place during encoding. Hence, encoding to the QRM code reduces to a classical error correcting code determined solely by the $Z$ stabilizer group of the original code. The remaining protocol is identical to the standard distillation scheme of Fig.~\ref{fig:distillation}.  
Ruiz~\emph{et al.}~\cite{ruiz2025unfoldeddistillationlowcostmagic} have shown that face stabilizers of the $\overline{QRM}(1,4)$ code that supports transversal $T$ can be placed on a planar chip using nearest-neighbour connectivity with only a few ancilla qubits, yielding a low-overhead and hardware-friendly alternative to costly logical-level schemes.

\begin{figure}[t]
\includegraphics[width=0.95\columnwidth]{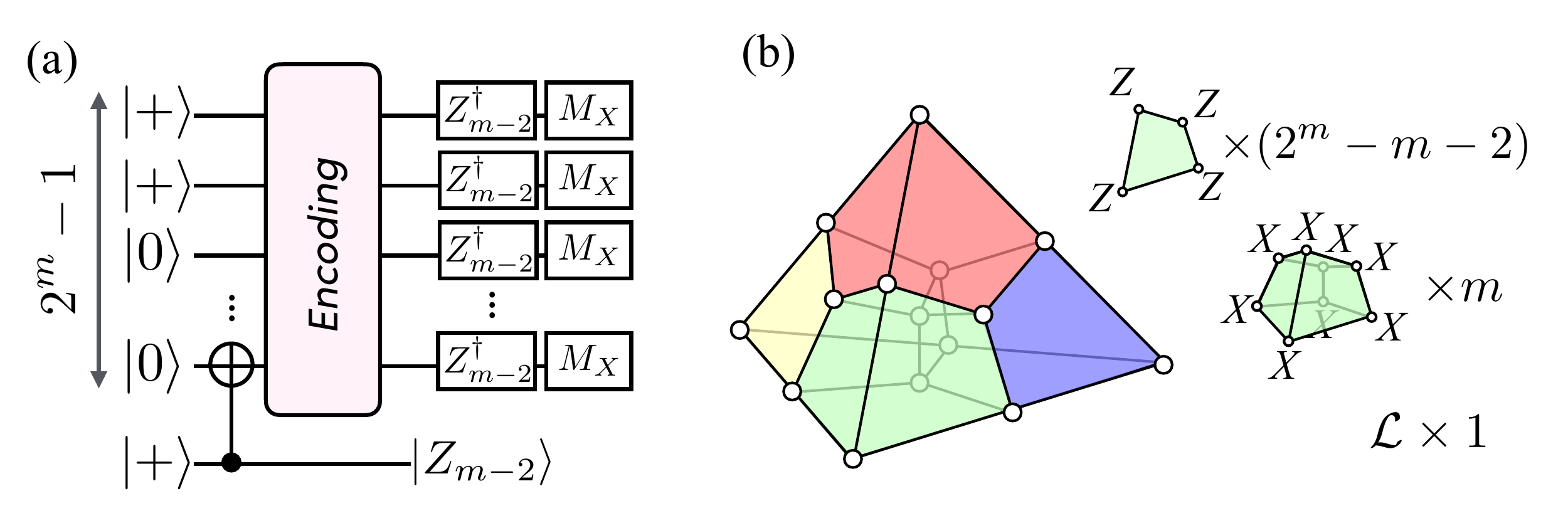}
\caption{\label{fig:distillation} Distillation of logical state $\ket{\tilde{Z}_{m-2}}$ using $\overline{QRM}(1,m)$ code. (a)~State distillation circuit consists of $2^m-1$ qubits of the code $\overline{QRM}(1,m)$ and one target qubit. First, the target is entangled with one of the code qubits in a Bell pair. The code qubits are then encoded in the $\overline{QRM}(1,m)$ code. 
Subsequently, $Z^{\dagger}_{m-2}$ gate are applied transversally to all qubits of the code, followed by $X$ measurement of all of the code qubits. The measurement leaves the target qubit in $\ket{Z_{m-2}}$. The volume stabilizers are reconstructed from the measurements and used to detect errors. 
(b)~The encoding circuit of~(a) fixes $2^m-m-2$ face and $m$ volume stabilizers. The encoding circuit can hence be replaced with direct stabilizer measurements. The shown example is the $\overline{QRM}(1,4)$ code that supports transversal $T$ gate and can be embedded locally in three-dimensional space. }
\end{figure}



In the following section, we develop the methodology for arbitrary $\overline{QRM}(1,m)$ codes and provide a constructive scheme for unfolding such codes onto a planar layout with very little qubit overhead, paving a way to resource-efficient distillation of higher-order non-Clifford gates $Z_k$. Since unfolded layouts constitute instances of parity codes, we refer to the technique as parity unfolding. 

\section{Parity-unfolded distillation}\label{sec:unfolding}

As outlined above, $\overline{QRM}(1,m)$ supports transversal implementation of the $Z_{m-2}$ gate. Here, we construct the parity-unfolded code, a classical error correcting code with weight-4 stabilizers equivalent to the code generated by face stabilizers of the $\overline{QRM}(1,m)$ code. This will be done in three steps: first, we fix the bulk of the unfolded code in Sec.~\ref{sec:bulk-structure}. Second, we remove the excessive degrees of freedom by adding boundary stabilizers in Sec.~\ref{sec:boundary-structure}. Finally, we use the parity code formalism to show transversality of $Z_{m-2}$ gates in the constructed parity-unfolded codes in Sec.~\ref{sec:parity-code}. Parity formalism is also employed to reconstruct volume stabilizers and correct errors. Finally, in Sec.~\ref{sec:nearest-neighbour} we show how the scheme can be realized in planar chips with only nearest-neighbour interaction at the cost of adding only a small number of physical qubits. 

\subsection{Bulk stabilizers}\label{sec:bulk-structure}

A shortened code $\overline{QRM}(1,m)$ encodes one logical qubit in $2^m-1$ physical qubits. It has $2^m - m - 2$ independent $Z$ stabilizers and $m$ independent $X$ stabilizers. The key property used in parity unfolding is that, although $\overline{QRM}(1,m)$ is local in at least $(m-1)$-dimensional space, all $Z$ stabilizer generators can be reduced to weight-4 for any $m$. Weights of $X$ stabilizers, on the other hand, grow exponentially with $m$, with each $X$ stabilizer having support on half of the QRM code qubits. 
However, since noise with a strong $X$-bias is considered, measuring volume $X$ stabilizers is not required during encoding. 

Before constructing parity-unfolded layouts, we note that the target qubit of the protocol in Fig.~\ref{fig:distillation} is itself, in principle, noisy and has to be protected, which can be done by encoding the target qubit either in the repetition code~(at infinite bias), or in a topological bias-tailored code~(e.g., thin surface, XZZX~\cite{XZZX} or $\textrm{X}^3\textrm{Z}^3$~\cite{PhysRevLett.133.110601} codes). In this section, we skip this complication and consider an ideal target qubit to focus solely on unfolding. A full scheme with noisy ancillas is then constructed by replacing the physical target qubit with an encoded qubit, which has no effect on the unfolding procedure. A scheme with noisy target qubits under biased noise is represented Fig.~\ref{fig:bulk} and analysed numerically in Sec.~\ref{subsec:numerical}. 

\begin{figure}[t]
\includegraphics[width=0.9\columnwidth]{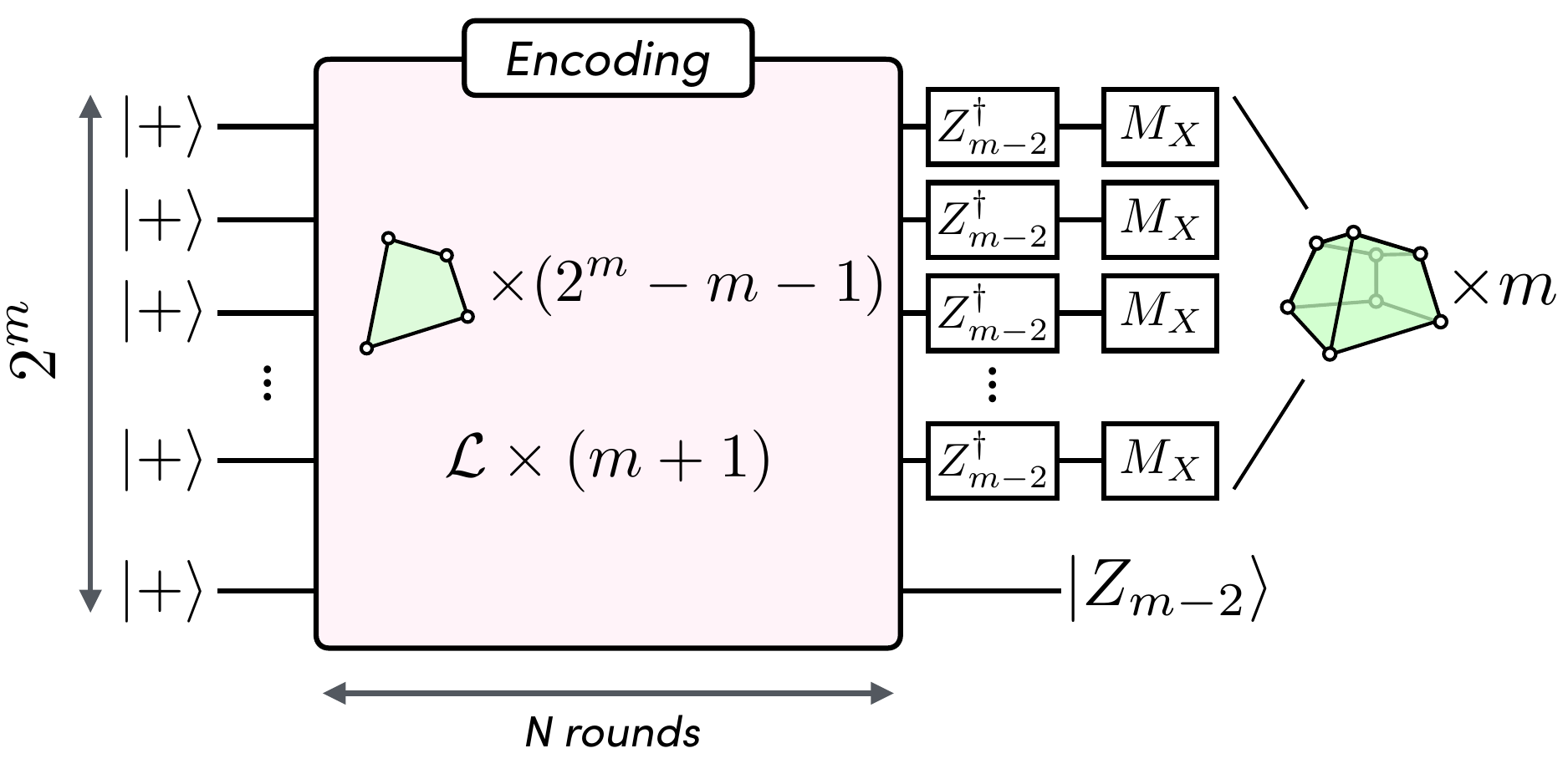}
\caption{\label{fig:unfolded-distillation} Parity-unfolded version of the distillation protocol in Fig.~\ref{fig:distillation}. Encoding to the QRM code is achieved by measuring stabilizers rather than executing the encoding circuit. Furthermore, since noise is strongly biased towards Pauli-$X$, only $Z$~(face) stabilizers are measured during encoding. Omitting measurement of $m$ volume promotes them to logical operators, resulting in $m$ additional encoded logical qubits. Secondly, the target qubit is entangled with the shortened QRM code by measuring an additional face stabilizer; it becomes part of the unfolded code and is included into the encoder box. Since one logical degree of freedom is added by including the target qubit to the code and one is removed by measuring an additional face stabilizer, the total number of logical qubits is unchanged. Finally, stabilizer measurements need to be performed $\nr$ times to correct measurement errors. 
}
\end{figure}

The usual description of the protocol provided in the previous section separates the system into a shortened Reed--Muller code $\overline{QRM}(1,m)$ and a target qubit. The target qubit must be maximally entangled with the logical qubit of the code, which can be done by measuring an additional weight-4 face stabilizer between the target and the code qubits. An equivalent description can be obtained by viewing the composite system consisting of $\overline{QRM}(1,m)$ together with the target qubit as the full (unshortened) $QRM(1,m)$ code, with one of the volume stabilizers promoted to a logical operator. As described in the Sec.~\ref{sec:qrm}, adding one qubit and one independent stabilizer indeed transforms the shortened code into the full one. Therefore, face stabilizers of the composite $\overline{QRM}(1,m)$ and target qubit is simply a classical $RM(1,m)$ code with parameters 
\begin{equation}\label{eq:rm-1-parameter}
    [n,k,d]
    =
    [2^m, m+1, 2^{m-1}].
\end{equation} 
Since only the face stabilizers need to be measured at strong bias, the task of parity unfolding reduces to finding a planar embedding of the $RM(1,m)$, which will be used instead of the encoding circuit, as illustrated in Fig.~\ref{fig:unfolded-distillation}.

\begin{figure}[t]
\includegraphics[width=0.95\columnwidth]{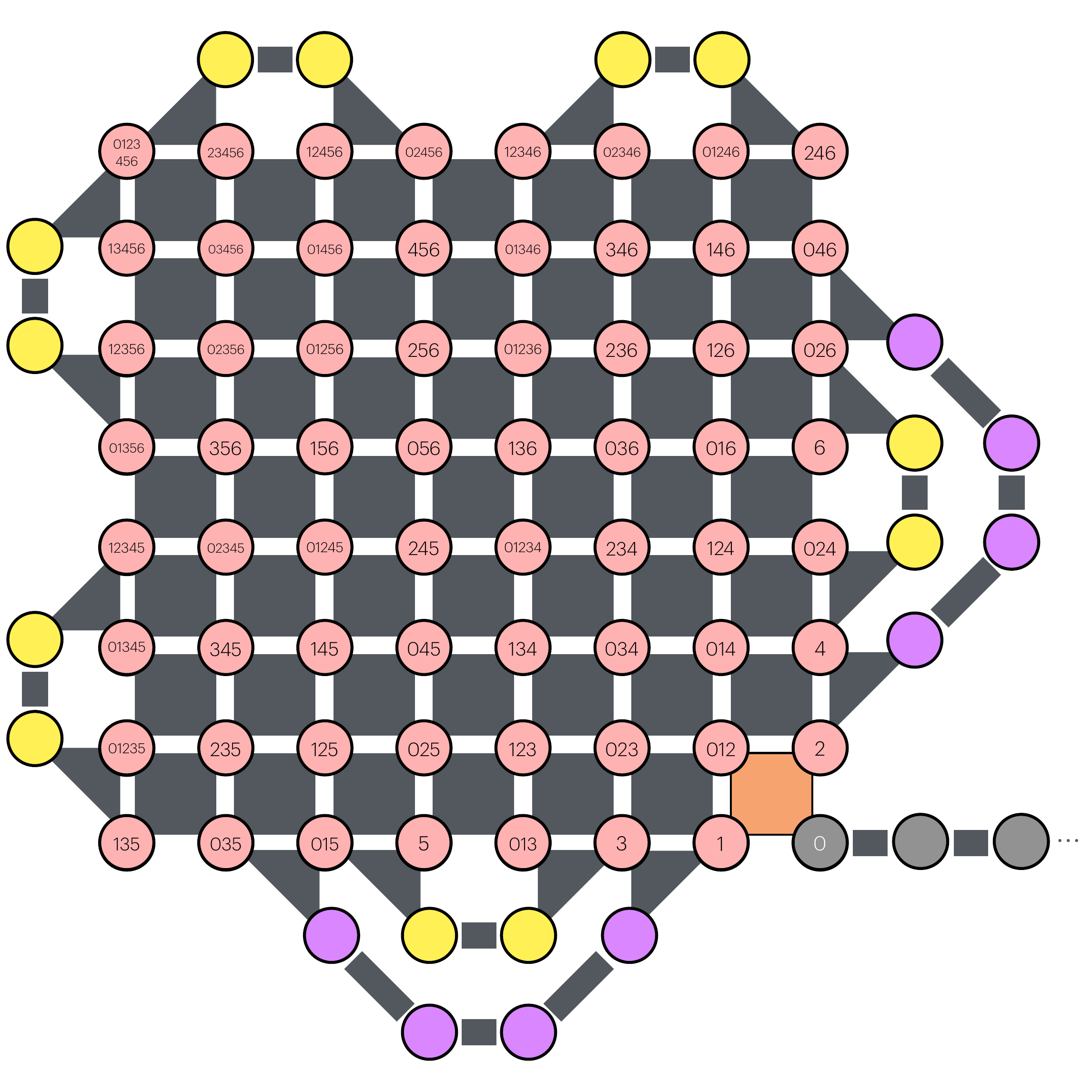}
\caption{\label{fig:bulk} A full layout for distilling rotation $Z_4$ using the unfolded code $uRM(6)$. The layout is constructed step by step, as described throughout Sec.~\ref{sec:unfolding}. There are $2^6$ bulk qubits~(circles with labels), of which $2^6-1$ are qubits of the $\overline{QRM}(1,6)$ code~(salmon) and one qubit of the target logical qubits~(labelled 0). Each square plaquette corresponds to a weight-4 $Z$ bulk stabilizer; there are $(2^3 - 1)^2 = 49$ independent bulk stabilizers. Construction of the bulk structure is described in Sec.~\ref{sec:bulk-structure}. The remaining $2^{m/2+1}-m - 2 = 8$ DOF are fixed by adding stabilizers $S(w,t)$ that have support on salmon qubits along the boundaries, as described in Sec.~\ref{sec:boundary-structure}. To make the layout compatible with the nearest-neighbour connectivity, each boundary stabilizer $S(w,t)$ is realized by distance-$w$ composite stabilizers. Here, each distance-2~(4) composite stabilizer consists of two yellow~(four purple) ancilla qubits and three~(five) stabilizers. Construction of the boundary stabilizers with nearest-neighbour connectivity is described in Sec.~\ref{sec:nearest-neighbour}. Each qubit is labelled using a parity label, as described in the proof of Lemma~\ref{lemma:2}. In case of even $m=2l$, there are $2l+1$ logical qubits and $2^{2l}$ parity qubits. For the bottom row, we create a pool of $l+1$ base labels and all possible $2^l - l - 1$ odd-length parity indices based on the base labels. We assign these indices to qubits of the bottom row according to the horizontal stabilizers. We do the same for the rightmost column. 
Parity labels of the remaining qubits are fixed from the rightmost column/bottom row by the bulk stabilizers step by step.  Logical $X$ operator $k$ in the unfolded layout corresponds to $X$ stabiliser $k$ and is given by the product of $X$ measurements of all parity qubits containing parity index $k$.
}
\end{figure}


Our construction relies on an important property of the first-order RM codes, namely, any linear code with parameters of Eq.~\eqref{eq:rm-1-parameter} is equivalent to the first order Reed--Muller code $RM(1,m)$. The proof is provided in Ref.~\cite{chen2009notesreedmullercodes}. Hence, any classical linear code with parameters of Eq.~\eqref{eq:rm-1-parameter} can be used for encoding under biased noise. Our goal is to find a constructive way to build such layouts on a planar chip with little qubit overhead. We will refer to such classical error correcting codes derived from the face stabilizers of the $\overline{QRM}(1,m)$ as unfolded Reed--Muller~(uRM) codes. For convenience we introduce the shortened and full uRM codes which we denote as $\overline{uRM}(m)$ and $uRM(m)$, respectively. The former corresponds to a classical code formed of face stabilizers of the shortened QRM code $\overline{QRM}(1,m)$, while the latter also includes the target qubit and stabilizers used to emulate the entangling CNOT gate in Fig.~\ref{fig:distillation}. When the target qubit is a single physical qubit, the full $uRM(m)$ must be equivalent to the $RM(1,m)$ code. 

We start by placing $2^m$ qubits on a $2^h \times 2^l$ squared lattice, with $l=h=m/2$ for even $m$ and $h=2^{(m-1)/2}$, $l=2^{(m+1)/2}$ for odd $m$. Without loss of generality, we first consider even $m=2l$, such as the one shown in Fig.~\ref{fig:bulk}. Weight-4 stabilizers are placed on the faces of the squared lattice, which we refer to as bulk stabilizers. Note that there are 
\begin{equation}
    \nzb
    = 
    (2^{l}-1)^2
    <
    \nz 
\end{equation}
bulk stabilizers. Hence, we need to fix the remaining $
\nz -\nzb = 2^{2l} - 2l - 1 - (2^{l}-1)^2 = 2^{l+1} - 2l - 2
$
independent weight-4 stabilizers, or, by symmetry, 
\begin{equation}\label{eq:side-stabilizers}
\begin{aligned}
    \nzs
    &=
    2^l - l - 1
\end{aligned}
\end{equation}
stabilizers along each of two types of boundaries, horizontal and vertical. We will refer to these missing stabilizers as boundary stabilizers.

\begin{figure}[t]
\includegraphics[width=0.85\columnwidth]{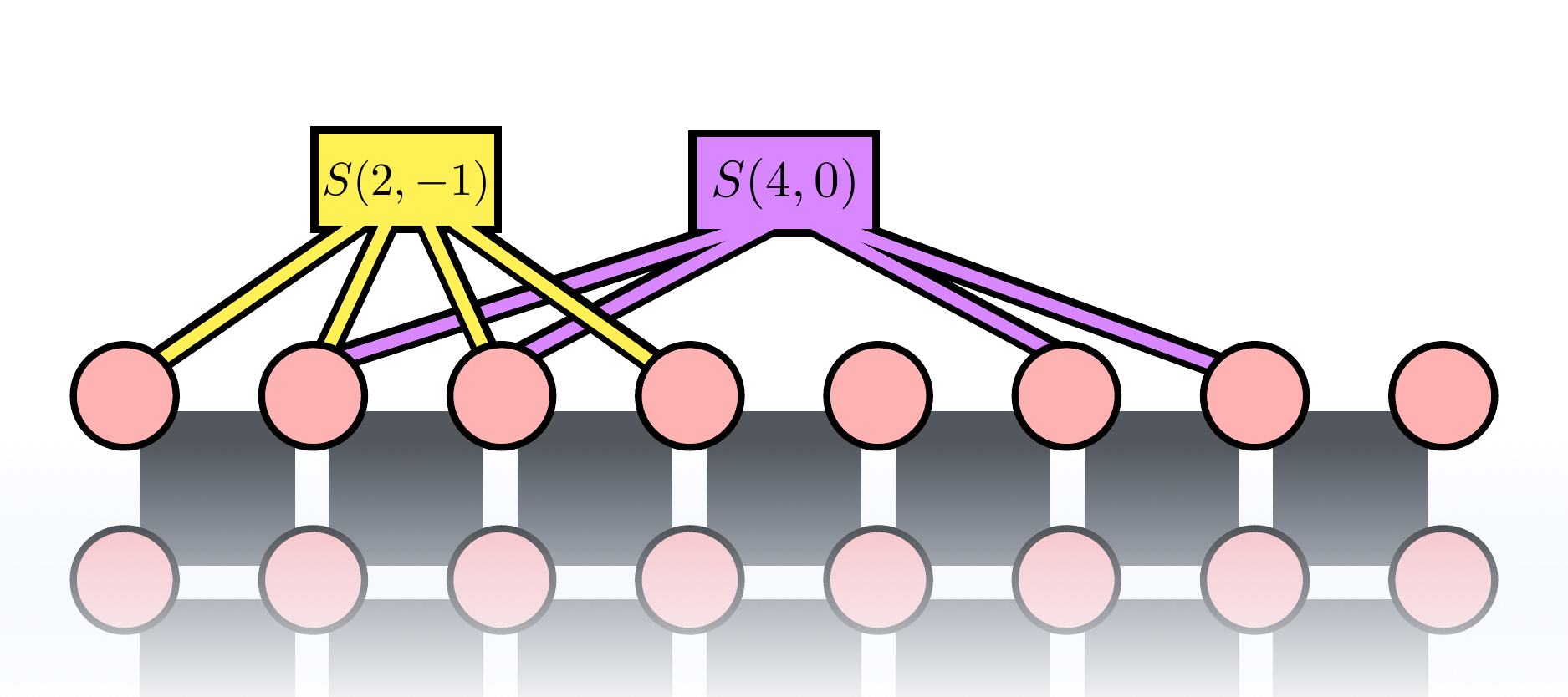}
\caption{\label{fig:edge} Examples of weight-4 stabilizers defined along the horizontal boundary of $uRM(4)$. Stabilizer $S(w,t)$ defined in Eq.~\eqref{eq:s-stabilizers} measures parity within two pairs of nearest-neighbour qubits. Distance $w$ of a stabilizers determines the distance between outer qubits of two pairs. Shift parameter $t$ determines the shift of the stabilizers to the left~($t<0$) and right~($t>0$) from the centre of the lattice. The two boundary stabilizers are realized using nearest-neighbour connectivity and shown in Fig.~\ref{fig:bulk} using the same colorscheme.}
\end{figure}

\subsection{Boundary stabilizers}\label{sec:boundary-structure}

We first fix boundary stabilizers along the horizontal boundary of the lattice. Stabilizers along the vertical boundaries can be derived by swapping row and column indices. There are many possibilities to fix the missing boundary stabilizers. Here, we provide the one we find especially simple due to its symmetry. Furthermore, as we show later, such choice of stabilizers can be realized with a small qubit overhead in planar chips with only nearest-neighbour connectivity. Consider the top row of the unfolded code $uRM(m)$, as shown in Fig.~\ref{fig:edge}. It has $\nq = 2^l$ qubits which we label 1 to $\nq$. Then, a boundary stabilizer we denote $S(w,t)$ will measure the $Z$ parity on sets of 4 qubits with indices 
\begin{equation}\label{eq:s-stabilizers}
    \begin{aligned}
        S(w,t)
        &=
        \Big{[}
        \frac{\nq}{2} - \frac{w}{2} + wt, \frac{\nq}{2} - \frac{w}{2} + wt + 1,
        \\&
        \frac{\nq}{2} + \frac{w}{2} + wt, \frac{\nq}{2} + \frac{w}{2} + wt + 1
        \Big{]},
    \end{aligned}
\end{equation}
where $w = 2^s$ with $s \in [1,l-1]$ and $t \in [-\nq / (2w) + 1, \nq / (2w) - 1]$. Parameters $w$ and $t$ in Eq.~\eqref{eq:s-stabilizers} admit simple geometrical interpretation. Each weight-4 stabilizer $S(w,t)$ measures Pauli-$Z$ on two pairs of qubits, left and right pairs. Each of the two pairs correspond to two nearest-neighbour qubits. Then, $w$ corresponds to a number of qubits between outer qubits of the two pairs. Parameter $t$ in $S(w,t)$ corresponds to a shift of qubits in $S(w,0)$ by $tw$ steps along the horizontal axis. Hence, we call $w$ and $t$ distance and shift parameters, respectively. Then the following theorem holds.
\begin{theorem}\label{theorem:1} 
Bulk stabilizers together with boundary stabilizers defined by Eq.~\eqref{eq:s-stabilizers} form the $RM(1,m)$ code.
\end{theorem}

The proof of this theorem relies on relies on Lemmas~\ref{lemma:1}, \ref{lemma:2}, \ref{lemma:3}, and the final proof is deferred to Sec.~\ref{sec:parity-code}.

\begin{lemma}\label{lemma:1} 
Bulk stabilizers together with boundary stabilizers defined by Eq.~\eqref{eq:s-stabilizers} form an independent set of $2^m - m - 1$ stabilizers on $2^m$ physical qubits, i.e., generate a code with parameters $[n,k]$ of Eq.~\eqref{eq:rm-1-parameter}, required for the $RM(1,m)$ code.
\end{lemma}
We prove Lemma~\ref{lemma:1} in Appendix~\ref{app:lemma-1}. In short, the Lemma shows that the number of independent stabilizers defined by~\eqref{eq:s-stabilizers} along the horizontal boundary of the code lattice is exactly $\nzs$ of Eq.~\eqref{eq:side-stabilizers}, so the total number of independent stabilizers is $(2^{l}-1)^2 + 2\nzs = 2^{m}-m-1$. Since the proof applies independently to horizontal and vertical boundaries of the lattice, it is also valid for the case of odd $m$, with replacing $l$ with $h$ in Eqs.~\eqref{eq:side-stabilizers} and \eqref{eq:s-stabilizers}.

\subsection{Equivalence of $RM(1,m)$ and $uRM(m)$}\label{sec:parity-code}

The unfolded code $uRM(m)$ constructed above constitutes a classical LDPC code with the required parameters $[n,k]$ of Eq.~\eqref{eq:rm-1-parameter}. 
To show that the code $uRM(m)$ is in fact identical to the $RM(1,m)$ code, we still require to show that the distance of all logical qubits of $uRM(m)$ is $2^{m-1}$. For that, we employ the parity framework~\cite{PhysRevLett.129.180503}. We also use the parity code techniques to reconstruct the volume stabilizers used in post-selection. 

A classical error correcting code with $Z$ stabilizers only protects against bit-flip errors, hence, each logical $Z$ operator has support on a single physical qubit. Since the $uRM(m)$ code encodes $(m+1)$ logical qubits, we can pick any $(m+1)$ physical qubits on which a single-qubit $Z$ operator corresponds to a logical single-qubit $Z$. In the parity code language, such physical qubits are referred to as base qubits. Physical $Z$ operators applied to any of the remaining $2^m - m - 1$ physical qubits correspond to multi-qubit logical $Z$ operators. These physical qubits are referred to as parity qubits. To each qubit of the code we assign a parity label, which denotes a subset of logical qubits, such that a $Z$ operator on the physical qubit translated to a product of logical $\tilde  Z$ operators on qubits of the subset. As such, to each base qubit we assign only a parity label corresponding to one logical qubit and to each parity qubit we assign parity label corresponding to a few logical qubits. Then the following lemma holds.

\begin{lemma}\label{lemma:2} 
The parity labels of qubits in $uRM(m)$ code are all possible odd-length combinations of $m+1$ logical indices.
\end{lemma}

\begin{lemma}\label{lemma:3}
Distance of each logical operator in the $uRM(m)$ code is $2^{m-1}$.
\end{lemma}
The proof of Lemmas~\ref{lemma:2} and \ref{lemma:3} is provided in Appendices~\ref{app:lemma-2} and \ref{app:lemma-3}, respectively. The proof of Theorem~\ref{theorem:1} follows directly from Lemmas~\ref{lemma:1}--\ref{lemma:3}.

\textbf{Proof of Theorem~\ref{theorem:1}.}
Each of $(m+1)$ logical qubits of the $uRM(m)$ codes has distance $2^{m-1}$, yielding parameters~\eqref{eq:rm-1-parameter}. Due to the uniqueness property of the first-order RM codes, the parity-unfolded $uRM(m)$ code is equivalent to the $RM(1,m)$ code. \qed

Since in the $QRM(1,m)$ code, $m$ logical operators of the $RM(1,m)$ code become stabilizers, parity formalism allow to reconstruct each volume stabilizer in Fig.~\ref{fig:unfolded-distillation}. 


\begin{remark}
    As a remark, we note that transversality of the $Z_{m-2}$ gate in the $uRM(m)$ can be derived solely from the parity formalism and without referring to the QRM codes. In Appendix~\ref{app:transversality-criteria}, we introduce a new gate transversality criteria that we call $k$-parity of the parity code and show that it is equivalent to the $k$-orthogonality criteria.
\end{remark}

\subsection{Nearest-neighbour connectivity}\label{sec:nearest-neighbour}

Equation~\eqref{eq:s-stabilizers} describes the complete set of boundary stabilizers. However, it does not explicitly describe a physical realization compatible with nearest-neighbour qubit connectivity, which is not a trivial task. As such, stabilizer $S(4,0)$ shown in Fig.~\ref{fig:edge} has support on qubits blocked by stabilizer $S(2,-1)$. In this section, we show how to realize stabilizers $S(w,t)$ on a planar layout with only nearest-neighbour connectivity and small qubit overhead. 

We start by noting that stabilizers measured along horizontal boundaries can be measured on qubits belonging to either the top or the bottom boundary of the code, since they are equivalent up to a product of bulk stabilizers along columns. Hence, approximately only a half of the required stabilizers need to be measured along one boundary. As such, in Fig.~\ref{fig:bulk}, two stabilizers are measured along the top boundary and two along the bottom one. The same applies to vertical stabilizers. 

As we climb up the Clifford hierarchy, distance $w$ of a certain small subset of stabilizers $S(w,t)$ grows exponentially in $m$. As $w$ corresponds to a number of qubits between the outer qubits of $S(w,t)$, some long-range connectivity is required for a small number of the code qubits located along the boundary. To realize such long-range stabilizers, we add new layers of ancilla qubits around the boundaries of the code, with higher-distance stabilizers constructed withing layers further away from the bulk of the code. Schematically, this is illustrated in Fig.~\ref{fig:arches}~(a). Next, we introduce two constructions that allow to access parity of the required code qubits by ancilla qubits located in outer layers. We refer to such constructions as centre-blocking and side-blocking. As shown in Fig.~\ref{fig:arches}~(b), the former allows to measure the parity of each of outer pairs of stabilizer $S(w,t)$ by adding a new layer of ancillas, but the parity of the inner two qubits can not be accessed. Similarly, the latter allows to measure only parity of the inner pairs of qubits while blocking the parity of the outer pairs. Because parities of the inner and outer pairs of qubits within each stabilizer can be accessed from the top and bottom boundaries, pairs of qubits required to measure a new stabilizer $S(w,t)$ are always accessible at the outer layer of the qubit layout--either the top or the bottom one. 

By placing stabilizers $S(w,t)$ with larger $w$ to the outer layers of the code, each distance-$w$ stabilizer $S(w,t)$ can be measured using exactly $w$ ancilla qubits. Since a distance-$w$ stabilizer has support on qubits which are $w$ qubits apart, our construction yields the smallest possible~(albeit possibly not unique) configuration that measures all stabilizers of unfolded codes $uRM(m)$ with nearest-neighbour connectivity.

\begin{figure}[t]
\includegraphics[width=0.95\columnwidth]{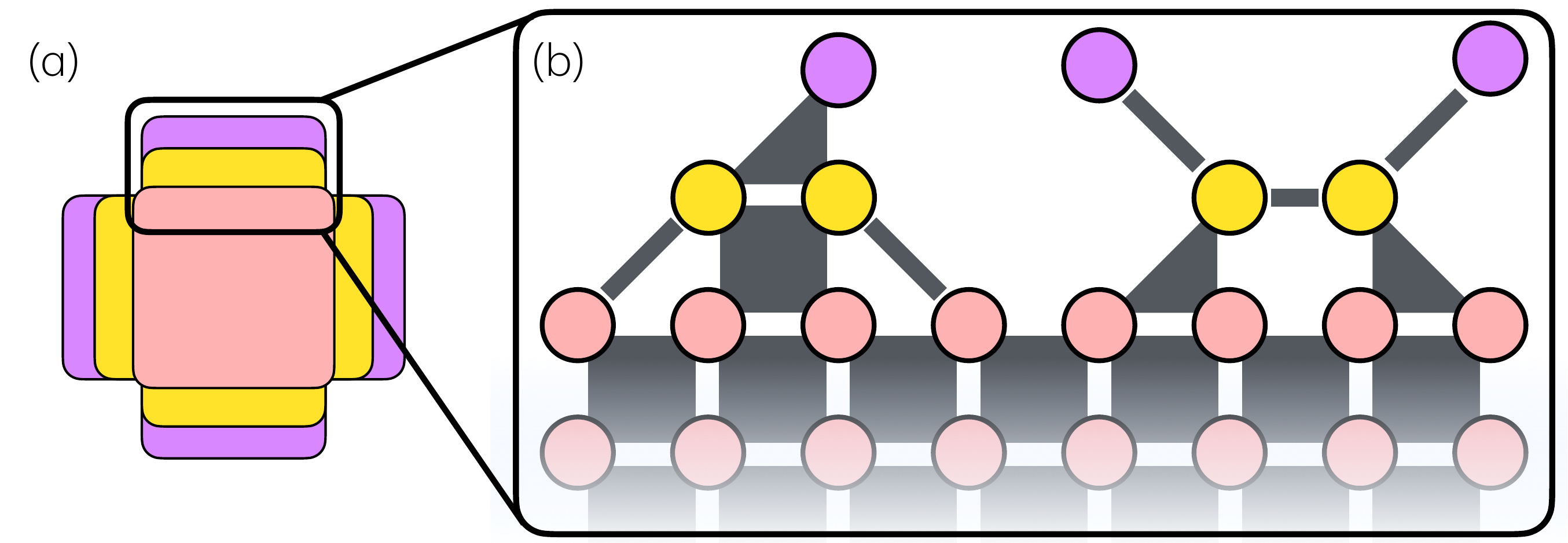}
\caption{\label{fig:arches} Construction of stabilizers $S(w,t)$ from Eq.~\eqref{eq:s-stabilizers} and Fig.~\ref{fig:edge} using nearest-neighbour connectivity. (a)~Schematic construction of qubit layout Salmon area corresponds to the code+target qubits. Yellow stripes corresponds to layers of ancilla qubits required to measure weight-2 stabilizers $S(2,t)$. Purples stripes corresponds to layers of ancilla qubits required to measure weight-4 stabilizers $S(4,t)$. To measure higher-distance stabilizers, more layers needed to be added. (b)~A zoom into a qubit structure for measuring boundary stabilizers. Shown are examples of two boundary distance-2 stabilizers $S(2,-1)$~(left) and $S(2,1)$~(right). Both stabilizers measure parity of 4 subsequent qubits and realized with nearest-neighbour connectivity. Each boundary stabilizer is measured by the means of two ancillas~(yellow) and three stabilizers. However, $S(2,-1)$ is measured using side-blocking construction. That is, the parity of the inner two qubits can be accessed from the next layer of the ancilla qubits~(purple), while parities of the side pairs are blocked. Similarly, the right construction used to measure $S(2,1)$ is centre-blocking, since only the parity of left and right pairs of qubits can be accessed from the top, while the parity of the inner pair is blocked.}
\end{figure}

Figure~\ref{fig:bulk} shows an explicit form of boundary stabilizers for $m=2l=6$, constructed as described above. As expected, the number of independent boundary stabilizers is $2^{l+1} - 2l - 2 = 8$, 6 of which are distance-2 and 2 and distance-4. Note also that our construction reproduces the unfolded $T$-distillation scheme of Ref.~\cite{ruiz2025unfoldeddistillationlowcostmagic} for $m=4$. In particular, the $2^4=16$ qubit layout requires $2^2 - 2 - 1 = 1$ stabilizer on the boundary of each type, described by $S(2,0)$. This stabilizer can be measured by using either a centre-blocking, or a side-blocking construction. 
In Appendix~\ref{app:nn-layout}, we provide explicit constructions of boundary stabilizers for distilling gates up to $Z_{k=8}$, i.e., up to $T^{1/64}$.

\subsection{Noisy target qubits}

Above we considered an idealized case with noiseless target qubit. In any real situation, the target qubit is itself noisy and has to be protected by an error correcting code. With infinite bias, we replace the target qubit with a repetition code, with the first qubit of the repetition code attached to the shortened $\overline{uRM}(m)$ code, which we refer to as interface qubit. The logical $Z$ operator on the repetition code translates to a single-qubit Pauli-$Z$ on the interface qubit. The situation is hence identical to the ideal target case, with the interface qubit acquiring a parity index 0, and the rest of the repetition code attached to it. This construction is illustrated in Fig.~\ref{fig:bulk}. 

Under finite noise bias, the target qubit has to be protected against both bit flips and dephasing. The shortened unfolded code $\overline{uRM}(m)$ in this case should be entangled with the target qubit on a logical level. As an example, when the surface code is used to protect the target logical qubit, the entangling CNOT gate of Fig.~\ref{fig:scheme} can be realised in a lattice-surgery manner. An example of a target qubit encoded in a rectangular surface code is illustrated in Ref.~\cite{ruiz2025unfoldeddistillationlowcostmagic}. In this case, the entire logical $Z$ operator of the surface code acquires the parity label 0. The discussion from the preceding sections remains valid, with the target qubit replaced by a logical qubit. We furthermore note that replacing the physical qubit with a logical one can only increase the code $X$-distance, hence, can only improve the protection against the dominant $X$ error during encoding.  

\section{Error analysis}\label{sec:error-analysis}

We now turn to the performance analysis of the parity-unfolded architecture. First, in Sec.~\ref{subsec:analytical}, we provide simple analytical consideration, in order to develop intuition about the qualitative behaviour of logical errors as we climb up the Clifford hierarchy. We then proceed to a full circuit simulation of the distillation experiment in Sec.~\ref{subsec:numerical}  

\subsection{Analytical analysis}\label{subsec:analytical}

Let us analyze how errors propagate in the distillation circuit of Fig.~\ref{fig:unfolded-distillation} and affect the fidelity of the output logical state. For simplicity, we again assume a noiseless target qubit, so only a single physical qubit is used to hold the distilled state in the end of the protocol.

We start by revisiting the distillation protocol of Fig.~\ref{fig:unfolded-distillation}. It begins with initializing $2^m$ physical qubits in a $(+1)$ eigenstate of Pauli-$X$ operator. This guarantees that a product of Pauli-$X$ operators on any subset of qubits is in the $(+1)$ eigenstate, i.e., volume stabilizers of the shortened $\overline{QRM}(1,m)$ code are fixed to $(+1)$. Face stabilizers of the $\overline{QRM}(1,m)$ code are fixed by measuring all stabilizers of the $uRM(m)$, which also entangles the encoded qubit with the target qubit. Since all Clifford elements, including CNOT gates used to measure stabilizers and ancilla measurements, are noisy, several rounds of syndrome extractions $\nr$ are needed to encode the state into the QRM codes with high fidelity. Subsequently, a layer of $Z^{\dagger}_{m-2}$ gates are applied to all qubits except the target qubit, which realizes a transversal logical $\tilde{Z}_{m-2}$ gate on the encoded qubit. Finally, all the qubits but the target are measured and volume stabilizers are reconstructed.

As introduced in Sec.~\ref{sec:qrm}, the $X$-distance of the $\overline{QRM}(1,m)$ is governed by the face stabilizers of the code and grows exponentially with $m$ as $d_x=2^{m-1}$. Here, we assume a noise level below the error correcting threshold of the parity-unfolded code $uRM(m)$. Furthermore, the number of syndrome extraction rounds $\nr$ is assumed to be constant for any considered $m$, but large enough to reduce the logical error rate of the unfolded code, governed by distance $d_x$ well below the total logical error rate of the distillation limited by distance-3 Pauli-$Z$ errors. In this setting, the logical error rate due to $X$ errors can be considered vanishingly small. 

Since the $Z$-distance of the code is governed by volume stabilizers and is fixed to $d_z=3$ for any $m$, weight-2 and 1 errors are always detectable. Some weight-3 Pauli-$Z$ configurations are undetectable, and cause a logical fault in the distilled state. A simple counting argument shows that there are $\binom{2^{k+2} - 1}{2}/3$ such undetectable combinations in the scheme used for distilling $Z_k$ gates, yielding the total distillation error rate
\begin{equation}\label{eq:distillation-ler}
\begin{aligned}
    P_{\textrm{dist}}(k)
    &=
    \frac{1}{3}
    \binom{2^{k+2} - 1}{2}
    p_{\textrm{eff},k}^3,
\end{aligned}
\end{equation}
where $p_{\textrm{eff},k}$ is the probability of a physical Pauli-$Z$ error affecting the final $X$ measurement when distilling the state $\ket{Z_k}$. Hence, behaviour of $p_{\textrm{eff},k}$ with $k$ determines the logical error rate of distillation. In the following, we will consider a regime where $p_{\textrm{eff},k}$ grows monotonically with $k$, as well as a regime where $p_{\textrm{eff},k}$ reduces with $k$ faster than the combinatorial term grows, leading to suppressed logical errors.

There are various sources of $Z$ errors affecting qubits before $X$ measurement. The first contribution is due to residual $X$ errors taking place just before the application of $Z_{k}$ gates. Since the gate $Z_{k}$ is not bias-preserving, it turns the residual $X$ error into a combination of $X$ and $Y$ errors, with the latter affecting the final $X$ measurement. Such a residual error can occur due to errors undetectable by stabilizers of the $uRM(m)$ code. In particular, a $X$ error on data qubits of CNOT gates used during the last round of syndrome extraction are not detectable by face stabilizers. Noise in each component of the unfolded code can be described by a Pauli channel
\begin{equation}\label{eq:1q-noise-model}
\begin{aligned}
    \mathcal{E}_1(\rho)
    =
    \sum_{P \in \{I,X,Y,Z\}}
    p_P
    P \rho P,    
\end{aligned}
\end{equation}
where $p_I$ is the gate fidelity and $p_P$ with $P \in \{X,Y,Z\}$ is the probability of the corresponding error. Similarly, a two-qubit gate between pair of qubits 1 and 2 is subject to noise described by a channel
\begin{equation}\label{eq:2q-noise-model}
\begin{aligned}
    \mathcal{E}_2(\rho)
    =
    \sum_{P_1,P_2 \in \{I,X,Y,Z\}}
    p_{P_1,P_2}
    P_1P_2 \rho P_1P_2.
\end{aligned}
\end{equation}
Under infinite bias, $p_X = p$ and $p_{IX} = p_{XI} = p_{XX} = p/3$, and other components are zero. With this, the probability of an undetectable error on the last-round CNOT is $p^{\textrm{CNOT}} = p_{IX} + p_{XX} = 2p/3$. The residual $X$ error is subsequently a subject to a $Z_k$ gate that transforms it as
\begin{equation}\label{eq:error-commutation}
    Z_k X Z_k^{\dagger}
    =
    \cos{\Big{(}\frac{\pi}{2^k}\Big{)}} X
    +
    \sin{\Big{(}\frac{\pi}{2^k}\Big{)}} Y,
\end{equation}
where the second term anti-commutes with the final $X$ measurement, hence contributing to the effective physical error rate $p_{\textrm{eff},k}$.

In addition to the residual errors, a noisy non-Clifford gate $Z_{k}$ produces Pauli noise described by a channel
\begin{equation}\label{eq:nc-noise-model}
    \mathcal{E}'_{k}(\rho)
    =
    \sum_{P \in \{I,X,Y,Z\}}
    q_P(k)
    P \rho P,
\end{equation}
where noise amplitudes $q_P(k)$ are, in principle, different from  the ones describing noise in Clifford gates. Unlike other components of the scheme, the non-Clifford gates does not have intrinsic noise bias. Hence, we consider symmetric depolarizing noise described by $q(k)\coloneqq q_X(k)=q_Y(k)=q_Z(k)$. It is, however, reasonable to assume an error probability scaling with the rotation angle $\pi/2^k$. Indeed, the duration of the pulse used to implement a $Z_k$ rotation scales as $1/2^k$ when Rabi frequency is kept constant, leading to the proportional scaling of accumulated noise. Alternatively, if errors during the gate occur predominantly due to non-adiabatic excitations, reducing Rabi frequency proportional to the rotation angle at constant pulse duration allow to significantly reduce non-adiabatic effects. This model agrees with experimental results in cat qubits. As such, Ref.~\cite{10.1038/s41467-017-00045-1} demonstrated reduction of noise approximately by a factor of 2 in $T$ gates compared to a $\pi/2$ rotation. Hence, we consider this model realistic at least for small values of $k$. In the following we will consider non-Clifford gates operating in two noise regimes. In a constant noise regime, the noise model does not depend on $k$,
\begin{equation}\label{eq:noise-constant}
    q{(k)}
    =
    q_X=q_Y=q_Z
    =
    \frac{p}{3}.
\end{equation}
In a scaled noise regime, the noise scales proportionally to the rotation angle $\pi/2^k$,
\begin{equation}\label{eq:noise-scaled}
    q{(k)}
    =
    q_X(k)=q_Y(k)=q_Z(k)
    =
    \frac{p}{3}
    \frac{1}{2^{k-2}},
\end{equation}
which is chosen such that for $T$ gates the noise is identical to the noise of Pauli and Clifford gates.

\begin{figure}[t]
\includegraphics[width=0.95\columnwidth]{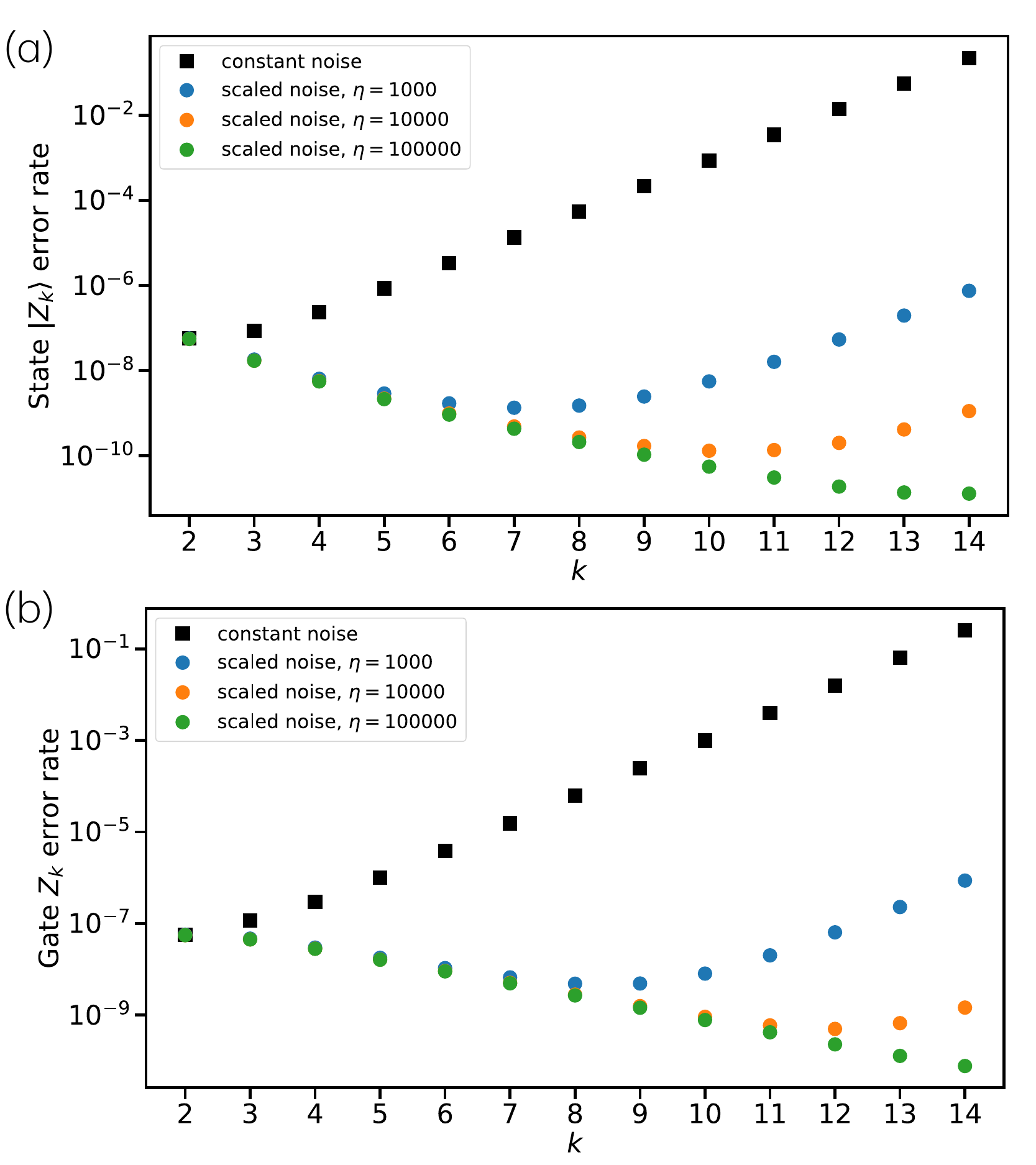}
\caption{\label{fig:ler}~Logical error rates for different values of bias $\eta$. (a)~Logical error rate of a single-round distillation of state $\ket{Z_k}$ calculated using Eqs.~\eqref{eq:distillation-ler} and \eqref{eq:p-eff}. Curves for two different noise models are shown. Squares correspond to a constant noise model of Eq.~\eqref{eq:noise-constant}, i.e., a situation where physical $Z_k$ rotations yield identical noise for all $k$. Since in this regime non-Clifford noise $q(k)$ dominates over other terms in Eq.~\eqref{eq:p-eff},  dependence on $\eta$ is unnoticeable and only a single curve is shown. Dots correspond to a scaled error model defined in Eq.~\eqref{eq:noise-scaled}.
(b)~Logical error rate of a full sequence of gates required to implement logical rotation $\tilde{Z}_k$ deterministically, calculated according to Eq.~\eqref{eq:cisc-seq-error}. In both plots, we use $p=10^{-3}$ and $\nr=10$.}
\end{figure}

Finally, in the finite-bias regime, a low-rate error can occur during $Z$ stabilizer measurements of parity-unfolded code with probability $\approx p/\eta$, where $\eta = p_Z/(p_X+p_Y)$ is noise bias. These errors accumulate during $\nr$ rounds of syndrome extraction and commute through the layer of non-Clifford gates unchanged, contributing to the total effective $Z$ noise at the final measurement. The total probability of a $Z$ error prior to $X$ measurement in a single physical qubit then reads
\begin{equation}\label{eq:p-eff}
    \begin{aligned}
    p_{\textrm{eff},k}
    &=
    p
    \Big{[}
    \frac{2}{3}
    \sin^2{\Big{(}\frac{\pi}{2^k}\Big{)}}
    +
    \frac{\nr}{\eta}    
    \Big{]}
    +
    2q(k),
\end{aligned}
\end{equation}
with $q(k)$ defied in either Eq.~\eqref{eq:noise-constant} or Eq.~\eqref{eq:noise-scaled}, depending on a noise regime. The factor 2 in the last term accounts for 2 out of 3 components of $q(k)$ anti-commuting with the $X$ measurement.

The two noise regimes result in qualitatively different scaling of the distillation fidelity with $k$. With a constant noise model, the logical error rate grows monotonically with $k$ due to a larger number of harmful weight-3 configurations in Eq.~\eqref{eq:distillation-ler}. That is, the fidelity of parity-unfolded distillation degrades as we climb up the Clifford hierarchy. A similar result has been previously demonstrated in other works on higher-level Clifford distillation~\cite{landahl2013complexinstructionsetcomputing}. However, we note that even this result should not be discouraging. As we will show in the following sections, even in this worst-case regime, the total error of approximating small-angle rotations using extended gate sets can be reduced due to shorter gate sequences required for synthesis.  

When noise of non-Clifford gate scales proportionally to the rotation angle, as in Eq.~\eqref{eq:noise-scaled}, we observe a qualitatively different behaviour. The logical error rate first decreases with $k$ until it achieves minimum at $k=k_{\textrm{th}}$. In this regime, suppression of the effective error rate rate in Eq.~\eqref{eq:p-eff} dominates over the growths of the number of harmful weight-3 error paths in Eq.~\eqref{eq:distillation-ler}. Eventually, the effective physical error rate saturates at $p\nr/\eta$, while the number of combinations in Eq.~\eqref{eq:distillation-ler} continues growing, leading to the increase in the logical error rate. Stronger bias $\eta$ shifts $k_{\textrm{th}}$ to higher values. 

Distillation of state $\ket{Z_{k}}$ itself is not the end goal of state injection protocols. The distilled state is subsequently consumed in a quantum teleportation circuit that applies the rotation $Z_{k}$ to the qubits of the data block, as explained in Fig.~\ref{fig:scheme}. Gate teleportation is inherently probabilistic and half of the time will produce rotation $\tilde{Z}^{\dagger}$ instead of $\tilde{Z}$. However, the incorrect sign will be detected during teleportation, indicating that it has to be corrected, which can be done by teleporting a rotation belonging to the previous level of the Clifford hierarchy $\tilde{Z}_{k-1} = \tilde{Z}_{k}^2$. Hence, deterministic application of the gate $Z_k$ forms a sequence of states $Z_k, Z_{k-1}..$, where distillation of the state $\ket{Z_{k'<k}}$ is requried with a probability
\begin{equation}\label{eq:p-smaller-k}
    P(k \rightarrow k')
    =
    \frac{1}{2^{k-k'}}.
\end{equation}
Then the total logical error rate generated by a chain of distillations required to apply $Z_k$ is 
\begin{equation}\label{eq:cisc-seq-error}
    \begin{aligned}
    P_{\textrm{gate}}(k)
    &=
    \sum_{2 \leq j \leq k}
    \frac{P_{\textrm{dist}}(j)}{2^{k-j}}
    \leq
    2
    \Big{(}
    1
    -
    2^{1-k}
    \Big{)}
    \mathrm{max}_{j \leq k}
    \Big{(}
    P_{\textrm{dist}}(j)
    \Big{)}
    \\&<
    2\mathrm{max}_{j \leq k}
    \Big{(}
    P_{\textrm{dist}}(j)
    \Big{)},
\end{aligned}
\end{equation}
where the maximum is over all levels of the Clifford hierarchy below the target rotation. From the last inequality, the logical error rate of applying the distilled gate $Z_k$ is upper bounded by twice the error rate of the noisiest $Z_{k' \leq k}$. Since under the scaled noise model the logical error rate monotonically reduces with $k$ for $k \leq k_{\textrm{th}}$, logical gate error in this regime is upper-bounded by $2P_{\textrm{dist}}(2)$. Logical error rates of the full sequence of distilled gates required for rotation $Z_k$ are calculated numerically according to Eq.~\eqref{eq:cisc-seq-error} and shown in Fig.~\ref{fig:ler}~(b). As in the case of a single state distillation under scaled non-Clifford noise, the total logical error rate of the sequence reduce with $k$ until $p_{\textrm{eff},k}$ in Eq.~\eqref{eq:p-eff} achieves saturation due a finite bias. 

\subsection{Numerical analysis}\label{subsec:numerical}


We have constructed quantum circuits for unfolded distillation experiments and numerically simulated distillation of gates for the first few levels of the Clifford hierarchy. Simulations are performed under a standard circuit-level noise with preparation, measurement, and gate errors described by Eqs.~\eqref{eq:1q-noise-model}, \eqref{eq:2q-noise-model}, and \eqref{eq:nc-noise-model} using both constant~[Eq.~\eqref{eq:noise-constant}] and scaled~[Eq.~\eqref{eq:noise-scaled}] error models of the non-Clifford gates. 

\begin{figure}[t]
\includegraphics[width=0.95\columnwidth]{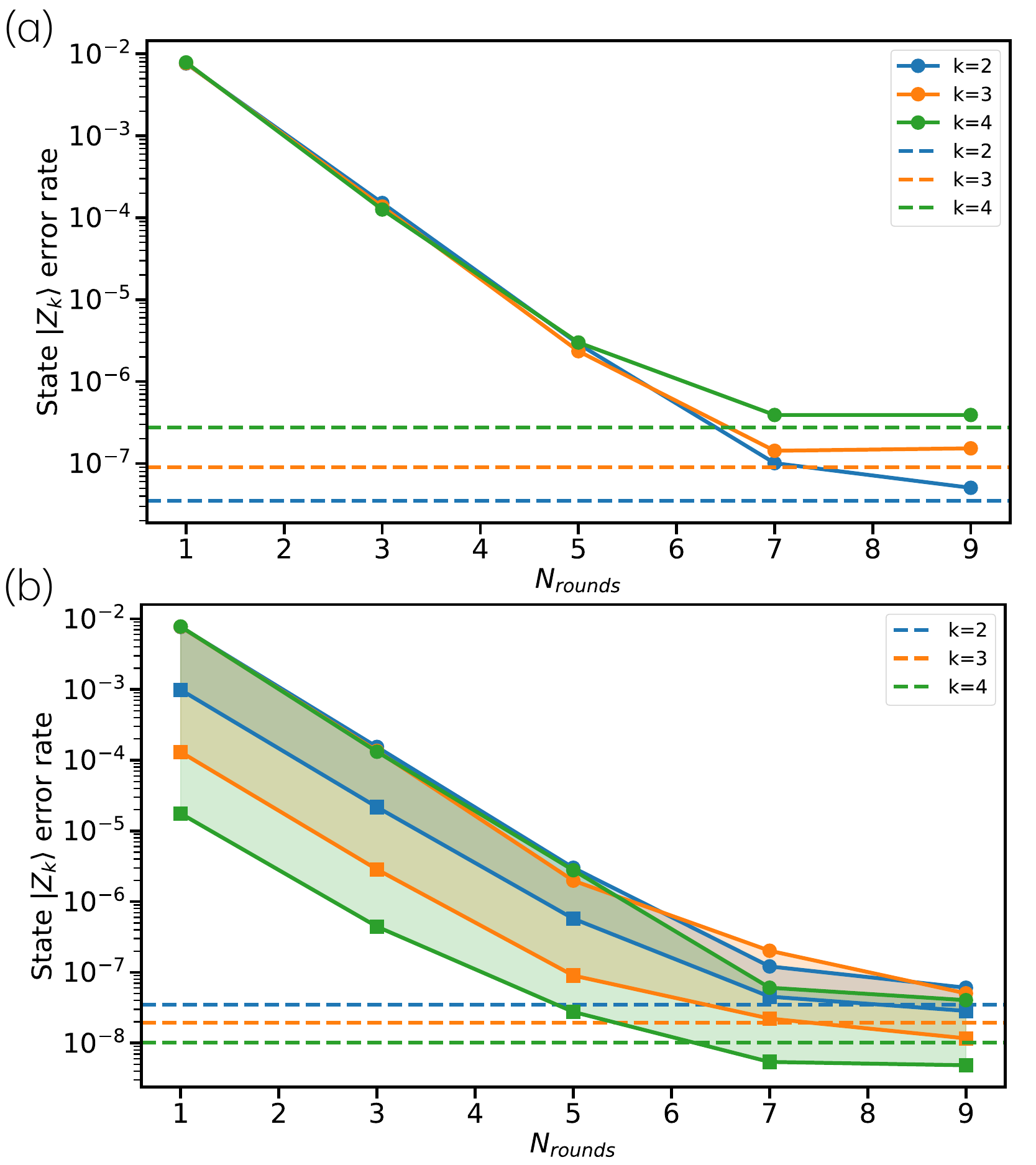}
\caption{\label{fig:ler-simulation}~Simulated logical error rates of the $Z_k$ state prepared using parity-unfolded distillation with (a)~constant and (b)~scaled non-Clifford noise. Dashed lines are the smallest theoretically achievable logical error rates calculated according to Eqs.~\eqref{eq:distillation-ler} and \eqref{eq:p-eff} and shown in Fig.~\ref{fig:ler}~(a). Since measurements of the $uRM(m)$ stabilizers are noisy, a minimum number of rounds is required to reach the best distillation logical error rate of Eq.~\eqref{eq:distillation-ler}. In (a), the curves, in fact, show the upper bound on the logical error rate, since our simulations overestimate the amount of noise due to $Z_k$ being replaced by $T$ gates. In~(b), we show both the upper- and lower-bound logical errors, as explained in the text. The true error rate lies in between the two bounds. Physical error probability of state preparation, measurement, and Clifford gates is assumed to be pure bit-flip with probability $p=10^{-3}$.}
\end{figure}

We use a few assumptions to simplify the analysis. Firstly, we simulate parity-unfolded layouts with long-range boundary stabilizers of Fig.~\ref{fig:edge} rather than composite boundary stabilizers of Fig.~\ref{fig:bulk}. Since the latter always yields larger code distances against Pauli-$X$ errors, the logical error rate derived in our simulations serves as an upper bound on the error rate achievable using a layout with composite stabilizers. Secondly, for all circuit elements except the non-Clifford gates, we assume a regime of infinite bias, so a repetition code of distance identical to the distance of the unfolded code can be used to protect the target qubit. Finally, we do one more modification to the parity-unfolded distillation circuit of Fig.~\ref{fig:unfolded-distillation} by executing half of the syndrome extraction rounds after the layer of non-Clifford gates. By doing so, we reduce the effect of measurement errors on the distilled state, since measurement errors preceding the non-Clifford gate will be corrected in the subsequent $\nr/2$ rounds of error correction. Without this, measurement error in the last rounds of syndrome extraction would go undetectable and lower-bound the achievable logical error rate. An alternative to this would be to use pre-selection used in Ref.~\cite{ruiz2025unfoldeddistillationlowcostmagic}. In a scheme with pre-selection, the measured syndrome is accepted when the measured syndrome does change within the $n_{\textrm{silent}}$ rounds of syndrome extraction. Such pre-selection discards experiments where errors take place just before the $Z_k$ gate, reducing the effect of measurement errors on the distilled state. However, the acceptance probability is suppressed exponentially with $k$, making the scheme incompatible with distillation of high-order rotations. 

Simulated logical error rates for the first few levels of the Clifford hierarchy are shown in Fig.~\ref{fig:ler-simulation}, with panels (a) and (b) corresponding to, respectively, the constant~[Eq.~\eqref{eq:noise-constant}] and scaled~[Eq.~\eqref{eq:noise-scaled}] noise of the non-Clifford gate. In agreement with the analytical results of Fig.~\ref{fig:ler}, the best achievable logical error rate increases with $k$ for the former and decreases for the latter. Since measurements of the $uRM(m)$ stabilizers are noisy, syndrome extraction needs to be repeated $\nr$ times, which plays the role of the temporal code distance. While the spatial distance of the parity-unfolded code $uRM(m)$
grows exponentially with $m$, $\nr$ needs not to be increased in a similar manner. Instead, we merely need to choose a minimum number of rounds $\nr$ required to reach the best distillation logical error rate of Eq.~\eqref{eq:distillation-ler}. Increasing $\nr$ beyond that does not suppress the distillation logical error rate any further. 

Our simulations show that the true minimum error rates of Eq.~\eqref{eq:distillation-ler} can be achieved using approximately the same number of syndrome extraction rounds $\nr = 9$ for all simulated values of $k$. In fact, for the constant non-Clifford noise model, saturation with higher $k$ is achieved even faster. This determines the choice $\nr=10$ made for all $k$ in the previous section. Importantly, this also means that temporal overhead of performing error correction on a parity-unfolded code does not grow with $k$. 

When the scaled noise model is used for the non-Clifford gates, Fig.~\ref{fig:ler-simulation}~(b) shows two curves for each value of $k$. 
Since simulating non-Clifford circuits is computationally intractable, we replace non-Clifford gates $Z_k$ with Clifford $T$ gates for each $k$. As a consequence, any $X$ error present in a qubit before the layer of $Z_k$ gates turns into $Y$ error with a probability of $1/2$ instead of $\sin^2(\pi/2^k)$, as it would in the real non-Clifford circuit, according to Eq.~\eqref{eq:error-commutation}. Hence, our simulations overestimate this type of error by a factor of $2/\sin^2(\pi/2^k)$, and yield an upper bound on the true logical error rate. We can extract the lower bound on the logical error rate by manually re-scaling the output logical error rate. To do that, we first extract the residual noise from the measured distillation error of Eq.~\eqref{eq:distillation-ler} as
\begin{equation}
    p_{\textrm{residual}}
    =
    p_{\textrm{eff},k}
    -
    \frac{2p}{3}\frac{1}{2^{k-2}}.
\end{equation}
We then re-scale it and use the new efficient $p_{\textrm{eff},k}$ in Eq.~\eqref{eq:distillation-ler},
\begin{equation}
    p_{\textrm{eff},k}
    =
    \frac{p_{\textrm{residual}}}{2^{k-2}}
    +
    \frac{2p}{3}\frac{1}{2^{k-2}}.
\end{equation}
As shown in Fig.~\ref{fig:ler-simulation}, the corresponding curve saturates slightly below the smallest expected logical error rate, hence being the lower-bound value. The true fidelity of the distillation is bounded within the shaded area in Fig.~\ref{fig:ler-simulation}, which includes the expected logical error rate. 

The code for constructing and simulating the performance of unfolded layouts, along with the simulation data, are available at~\cite{unfolding-simulation-repo}. Our implementation uses qLDPC Python package~\cite{perlin2023qldpc} for simulating stabilizer QECCs.

\section{Resource overheads}\label{sec:resource}

Finally, we compute resource overheads required for certain important subroutines with parity-unfolded architecture and benchmark them against overheads of the unfolded (Clifford + $T$) gate set. We denote $C_k$ a set of unfolded gates that contains all non-Clifford gates from the levels of the Clifford hierarchy up to the level $(k+1)$. Because we are focused on the distillation of non-Clifford gates, Clifford gates are considered freely available at no cost. With this, (Clifford+$T$) becomes $C_2$, (Clifford+$T$+$\sqrt{T}$) becomes $C_3$, etc.

Resource overhead is calculated in terms of space-time cost required to apply a gate, as given by the number of physical qubits times the circuit depth. Furthermore, we are interested in the cost of deterministic gate as a meaningful metric. That is, the cost of applying the rotation $Z_k$ to the logical qubit of the data block is given by the the cost of distilling $Z_k$ plus the cost of distilling all the corrective gates $Z_{k'<k}$ weighted by the probability~\eqref{eq:p-smaller-k}, with Clifford gates~($k'=1$) considered free. 

We will focus on two important use cases. In Sec.~\ref{subsec:cost-native}, we calculate the cost of executing fault-tolerant rotations $Z_k$, which emerge naturally in a number of important quantum computing subroutines, such as the Quantum Fourier Transform or phase estimation. As these rotations are native to our scheme, one can expect the biggest advantage in this case. In Sec.~\ref{subsec:cost-arbitrary}, we compute the cost of approximating an arbitrary single-qubit unitary $U$ using unfolded $C_2$ and extended $C_{k>2}$ sets, and observe that the latter allows to significantly reduce both the cost and total logical error rate.

\subsection{Native $Z_k$ rotations}\label{subsec:cost-native}

Let us calculate the total space-time cost of applying the gate $Z_{k}$. This gate is transversal in the code $\overline{QRM}(1,m)$ which we realize using parity-unfolded code $uRM(m)$ with $m=k+2$. The latter contains $2^{m}$ physical qubits, of which $2^m-1$ are qubits of the $\overline{QRM}(1,m)$ code and one target qubit. For simplicity, we assume ideal target qubit, so no additional encoding of the target is required. When the target logical qubit is noisy, it needs to be encoded in, e.g, a repetition code of the same distance as logical qubits of the QRM code, which is a half of the QRM code size. Hence, adding the encoded ancilla would increase the cost of distillation for each $k$ by a constant pre-factor of 3/2, therefore not affecting the relative cost between distillation of different $Z_k$ gates.

In addition to the $2^m$ bulk qubits of the $\overline{QRM}(1,m)$ code, we require qubits to measure all boundary stabilizers $S(w,t)$ in Eq.~\eqref{eq:s-stabilizers}. When realized using long-range connectivity along the boundary, such stabilizers do not require additional data qubits, and the total number of physical qubits including ancillas for stabilizer measurement in the $uRM(m)$ code is
\begin{equation}\label{eq:n-qubits-lr}
\begin{aligned}
    N^{(\textrm{lr})}_{\textrm{qubits}}(k)
    &=
    2^{k+3} - k - 3
    \approx
    2^{k+3},
\end{aligned}
\end{equation}
where we switched to variable $k=m-2$ denoting the rotation $Z_k$ .

With nearest-neighbour connectivity only, each boundary stabilizer is composite and requires $w$ additional data qubits, as shown in Fig.~\ref{fig:bulk} and explained in Sec.~\ref{sec:nearest-neighbour}. Assume again even $m=2l$ for simplicity. Recall from the proof of Lemma~\ref{lemma:1} that for each $w$ there are $\nq /w - 1 = 2^{m/2-s} - 1$ boundary stabilizers, each requiring $w=2^s$ data qubits. Then, for each of two types of boundaries, one needs
\begin{equation}
    \begin{aligned}
        \sum_{s=1}^{m/2-1}
        (2^{m/2-s} - 1)2^s
        =
        \Big{(}
        \frac{m}{2}-2
        \Big{)}
        2^{m/2} + 2
    \end{aligned}
\end{equation}
additional data qubits. With nearest-neighbour connectivity, $uRM(m)$ hence contains
\begin{equation}\label{eq:n-data}
    N^{(\textrm{nn})}_{\textrm{data}}(k)
    =
    2^{k+2}
    +
    \Big{(}
    k-2
    \Big{)}
    2^{k/2+1}
    +
    4
\end{equation}
data qubits. The total number of qubits in the $uRM(m)$ code with nearest-neighbour connectivity is then
\begin{equation}\label{eq:n-qubits-nn}
\begin{aligned}
    N^{(\textrm{nn})}_{\textrm{qubits}}(k)
    &=
    2
    N_{\textrm{data}}(k)^{(nn)}
    -
    k
    -
    3
    \\&=
    2^{k+3}
    +
    (k-2)2^{k/2+2}
    -
    k
    +
    5
    \\&\approx
    2^{k+3} + k2^{k/2}
\end{aligned}
\end{equation}
In the following, $N_{\textrm{data}}(k)$ will refer to either Eq.~\eqref{eq:n-qubits-lr} or Eq.~\eqref{eq:n-qubits-nn} depending on our assumption on the available connectivity along the boundary of the code patch. Intermediate regimes are possible, e.g., when next-neighbour connectivity is available. However, in any case, all such configurations only differ by extra qubits along the code boundary, and all yield identical qubit overhead in the leading order.  

\begin{figure}[t]
\includegraphics[width=0.95\columnwidth]{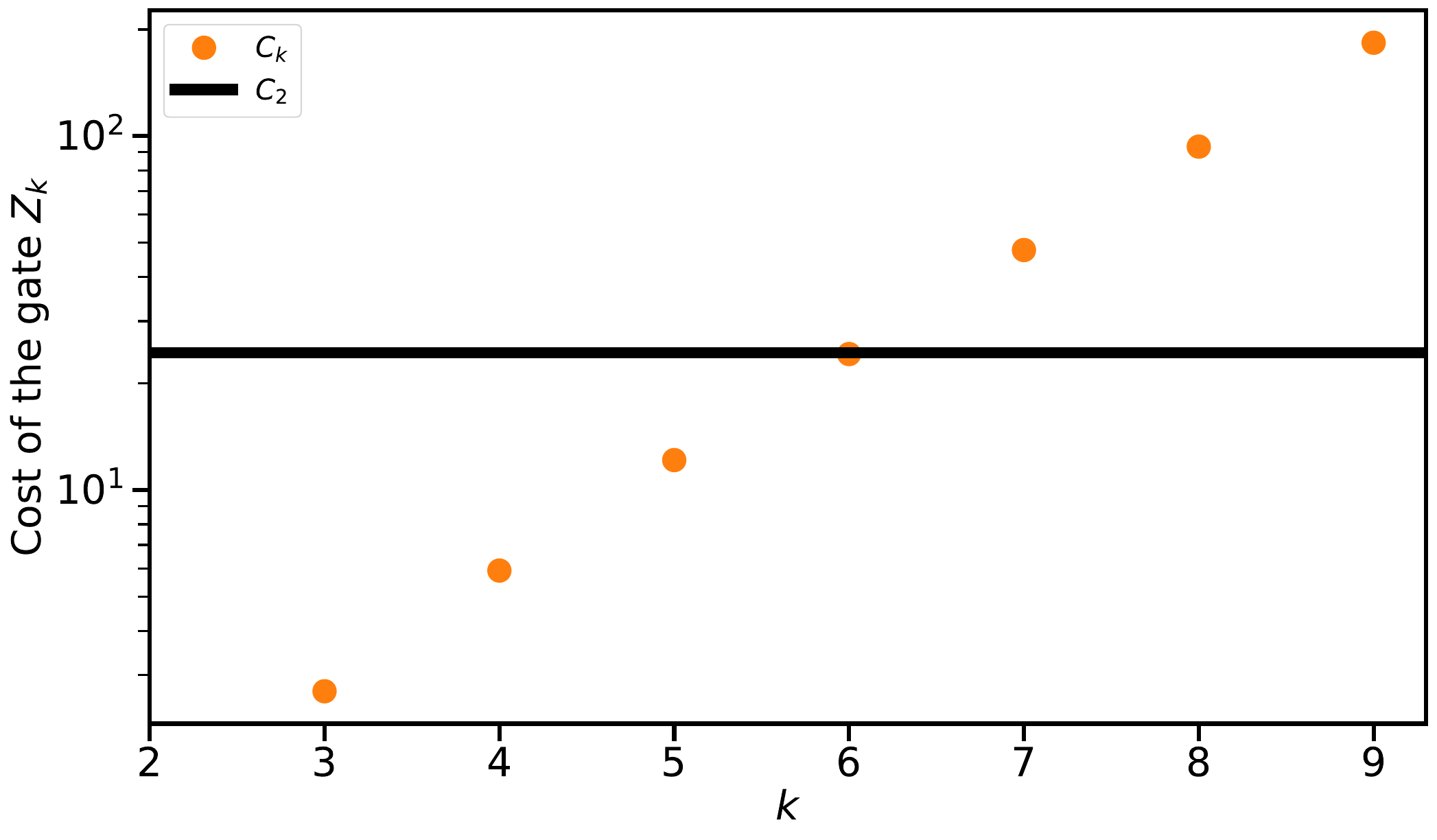}
\caption{\label{fig:cost} Space-time overheads of applying a $Z_k$ rotation using parity-unfolded $C_k$ gate set. Orange circles correspond to the exact and upper-bound expression Eq.~\eqref{eq:cisc-cost}. The horizontal line shows the cost of approximating $Z_k$ with gates from $C_2$, i.e., the (Clifford + $T$) set, calculated according to Eq.~\eqref{eq:risc-cost}. Here, we use $\eta = 10^5$, $\epsilon=10^{-3}$ and the rest or parameters are identical to those of Fig.~\ref{fig:ler}. Space-time cost is normalized on the cost of distilling a single $T$ gate.}
\end{figure}

As in the standard distillation protocol, the distilled state is rejected  when one or two errors are detected. In this case, the experiment is discarded and the protocol is executed again from the beginning. In the leading order, the success probability reads $P_{\textrm{success}}(k) = 1- N_{\textrm{data}}(k)p_{\textrm{eff},k}$ 
with $p_{\textrm{eff},k}$ defined in Eq.~\eqref{eq:p-eff}. The number of distillation attempts before the state is accepted is 
\begin{equation}
\begin{aligned}
    N_{\textrm{attempts}}(k)
    \approx
    \frac{1}{1
    -
    N_{\textrm{data}}(k)p_{\textrm{eff},k}}
\end{aligned}
\end{equation}
Additionally, the space-time overhead increases proportionally to the number of syndrome extraction rounds $\nr$. Since, according to the results of the previous section, the smallest lower-bound logical error rate Eq.~\eqref{eq:distillation-ler} is reached with the same number of syndrome extraction rounds for any considered $k$, it adds the same pre-factor $\nr$ to the cost of distilling each gate $Z_k$ and can be omitted. 

With this, the total space-time cost of distilling the state $\ket{Z_{k}}$ in a parity-unfolded layout $uRM(m=k+2)$ is $N_{\textrm{qubit}}(k)N_{\textrm{attempts}}(k)$. Since we are interested in the cost of deterministic gate $Z_k$ with the correct sign, distillation of a gate $Z_{k'<k}$ might be required with probability given in Eq.~\eqref{eq:p-smaller-k}. Denote $R_k(k')$ the space-time cost of applying non-Clifford rotation $k$ using parity-unfolded gates from the set $C_{k'}$. The average resource cost of applying rotation ${Z}_k$ deterministically using an unfolded $C_k$ gate reads
\begin{equation}\label{eq:cisc-cost}
\begin{aligned}
    &R_k(k)
    =
    \sum_{2 \leq j \leq k}
    \frac{N_{\textrm{qubits}}(j)
    N_{\textrm{attempts}}(j)}{2^{k-j}}    
    \\&\leq
    \textrm{max}_{2 \leq j \leq k}
    \big{[}
    N_{\textrm{qubits}}(j)
    N_{\textrm{attempts}}(j)
    \big{]}
    \sum_{2 \leq j \leq k}
    \frac{1}{2^{k-j}}
    \\&<
    2
    N_{\textrm{qubits}}(k)
    N_{\textrm{attempts}}(k),
\end{aligned}
\end{equation}
where in the last inequality we again used $\sum_{2 \leq j \leq k}{1}/{2^{k-j}} \leq 2$.

\begin{figure}[t]
\includegraphics[width=1.0\columnwidth]{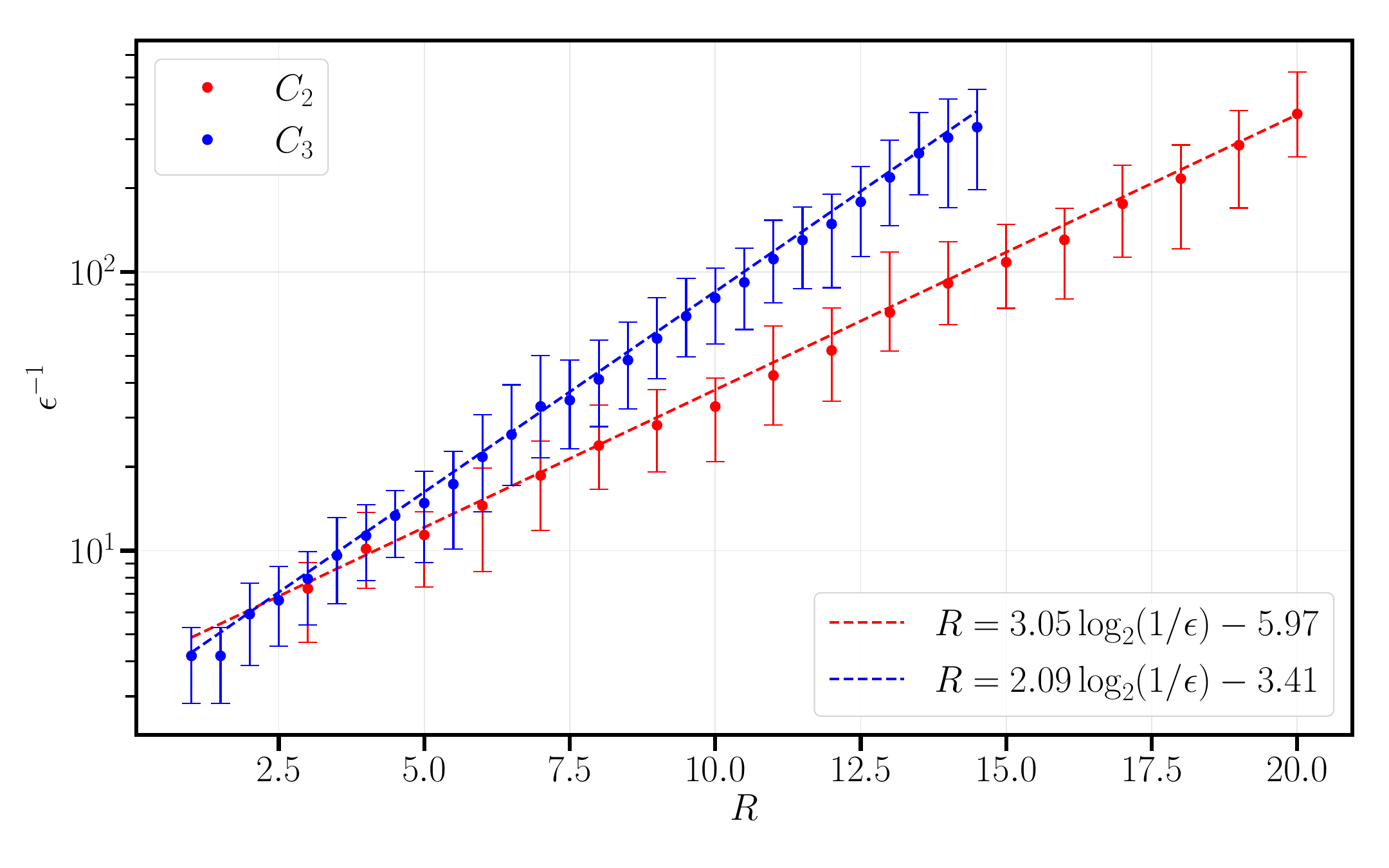}
\caption{\label{fig:synthesis} Average synthesis error $\epsilon$ as a function of total resource cost $R$ for the unfolded $C_2 = (\textrm{Clifford} + T)$ and $C_3 = (\textrm{Clifford} + T + \sqrt{T})$ gate set. Each data point displays the average synthesis error for 100 Haar random unitaries. Dashed lines indicate fits to $a\log(1/\epsilon) + b$. Data points mark the median while error bars indicate the $68\%$ confidence interval.}
\end{figure}

When the $C_2$ set is used to approximate $Z_{k=m-2}$ rotation, a sequence of Clifford and $T$ gates is required to achieve the required synthesis error $\epsilon$. The number of $T$ gates $N_{\textrm{gates}}(\epsilon)$ only depends on $\epsilon$ and is independent of $k$. Only Clifford corrections are required when $Z_{2}^{\dagger}$ is teleported instead of $Z_{2}$, and they are considered free. The total space-time overhead for applying the gate $Z_{k}$ with the $C_2$ set is therefore
\begin{equation}\label{eq:risc-cost}
    R_k(2)
    =
    N_{\textrm{gates}}
    N_{\textrm{qubits}}(2)
    N_{\textrm{attempts}}(2).
\end{equation}

For large $k$, the synthesis of the target gate $Z_k$ is more resource efficient than direct distillation, since distillation cost is independent of $k$, while the cost of distilling gates $Z_k$ grows fast with $k$. For relatively small $k$, however, direct distillation of $Z_k$ might be more efficient. Using Eqs.~\eqref{eq:cisc-cost} and \eqref{eq:risc-cost}, and solving $R_k(4) = R_k(k)$, we can make a rough estimation of break-even value $k_{\textrm{b}}$ such that for $k \leq k_{\textrm{b}}$, direct distillation of $Z_k$ using parity unfolding is cheaper that approximating it using unfolded $T$ gates. For moderate values of $k$, we use $N_{\textrm{attempts}}(2) \approx N_{\textrm{attempts}}(k) \approx 1$ and  $N_{\textrm{qubits}}(k) \approx 2N_{\textrm{qubits}}(k-1)$, which gives us the break-even value of 
\begin{equation}\label{eq:break-even}
    k_{\textrm{b}}
    \approx
    2 + \log_2{(N_{\textrm{gates}})},
\end{equation}
The only parameter here is the number of $T$ gates in a gate sequence, and it depends on the efficiency of used decomposition. 
Here, we consider an efficient decomposition we derive using \texttt{TRASYN}, a recently introduced method based on tensor networks \cite{Hao2025}, that yields
\begin{equation}\label{eq:number-of-T}
    N_{\textrm{gates}}(\epsilon)
    =
    3.05\log_{2}{{(} 1/\epsilon{)}} - 5.97.
\end{equation}
where $\epsilon$ is the trace distance between synthesised and target unitary. We choose the synthesis error such that $\epsilon^2 \approx N_{\textrm{gates}}(\epsilon) P_{\textrm{dist}}(2)$, that is, we want the synthesis errors approximately equal to the total error accumulated from noisy $T$ gates. Using $\epsilon = 10^{-3}$, Eq.~\eqref{eq:break-even} yields $k_{\textrm{b}} = 6.6$. That is, we can reduce the cost by directly distilling and teleporting all gates from the gate set $C_6$ rather than approximating them with Clifford and parity-unfolded $T$ gates. This rough estimate agrees well with the numerical calculations of $R_k(k)$ presented in Fig.~\ref{fig:cost}.

\begin{figure}[t]
\includegraphics[width=1.0\columnwidth]{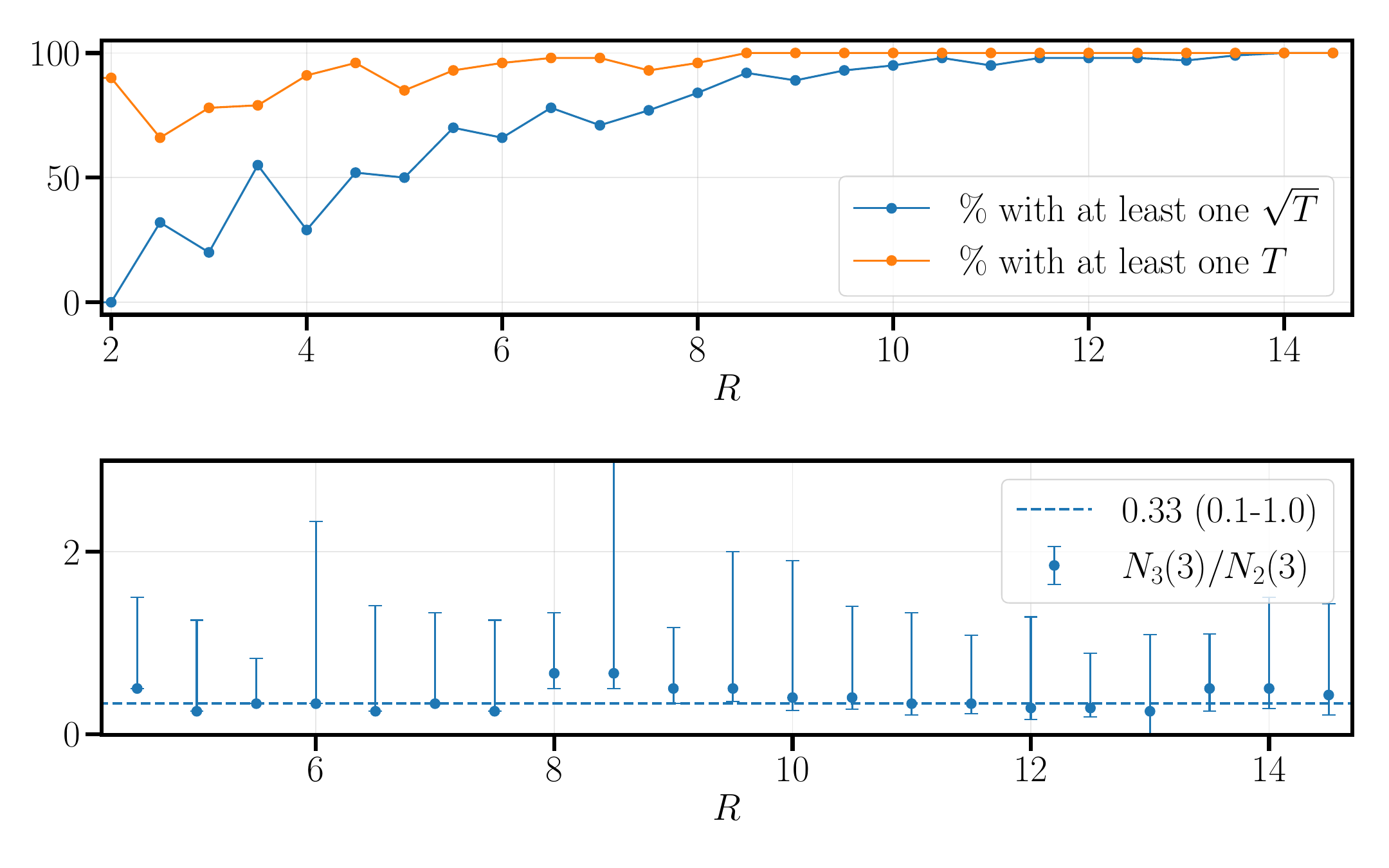}
\caption{\label{fig:analyse_sequences} Top panel: Percentage of synthesized unitaries $V$ which require at least one $\sqrt{T}$, at least one $T$, respectively. Bottom panel: Ratio of the number $\sqrt{T}$, $N_3(3)$, and $T$, $N_2(3)$, for each each synthesised sequence. Data points mark the median while error bars indicate the $68\%$ confidence interval.}
\end{figure}

\subsection{Arbitrary single-qubit unitary}\label{subsec:cost-arbitrary}

Consider now a random single-qubit unitary $U$. Can one save resources when synthesising it using the gate set $C_{k>2}$ instead of $C_{2}$? To answer this question, we require a decomposition algorithm that synthesises a desired unitary using gates from the extended set. 

The task is formulated as follows: For a given unitary $U$ and $k>2$,
fix a resource cost of gates from $C_k$, where Clifford gates are free and a gate cost $R_{k'}(k')$ for $2 \leq k' \leq k$ is given by Eq.~\eqref{eq:cisc-cost}. Given the total resource cost $R$, synthesise a unitary $V$ as a sequence of $Z_{k'} \in C_k$ gates and Clifford gates which minimizes the trace distance $\epsilon$ with respect to $U$

\begin{equation}
\label{eq:synthesis_minimisation}
    \min_{\{N_{k'}(k)\}_{k'<k}} 
    \epsilon
    {(}
    \{N_{k'}(k)\}
    {)} 
    \quad \text{s.t.} \quad \sum_{k' \leq k}
    N_{k'}(k) R_{k'}(k')
    \leq R,
\end{equation}
where $N_{k'}(k)$ is the number of $Z_{k'}$ gates. $\epsilon$ is the trace distance of Choi-states defined as 
\begin{equation}
\label{eq:synthesis_error}
    \epsilon = \frac{1}{2}\vert\vert U - V \vert\vert_\mathrm{tr} = \sqrt{1 - \frac{ \vert \mathrm{Tr}\left(VU^\dagger\right)\vert^2}{4}},
\end{equation}
where $\vert\vert A \vert \vert _{\mathrm{tr}} = \mathrm{tr}(\sqrt{A^\dagger A})$.

This problem is essentially a reformulation of the ancilla-free synthesis problem for a single qubit unitary $U$ where we aim to find the cheapest gate sequence (in terms of cost-weighted number of gates) that approximates $U$ within a fixed error budget $\epsilon$. Note that the trace distance $\epsilon$ is an upper bound for the diamond norm and related to the process fidelity $\mathcal{F}$ by $\epsilon^2 = 1 - \mathcal{F}$. 


For simplicity, we normalize the cost of each gate on the cost of distilling a $T$ gate, i.e., $R_2(2)=1$. The cost of approximating a random unitary with the $C_2$ gate set is simply the number of gates $N_{\textrm{gates}}$, as given by Eq.~\eqref{eq:number-of-T}. The normalised cost of $\sqrt{T}$ is then
\begin{equation}
    R_3(3)
    =
    \frac{N_{\textrm{qubits}}(3)
    N_{\textrm{attempts}}(3)}
    {N_{\textrm{qubits}}(2)
    N_{\textrm{attempts}}(2)}
    +
    1/2
    \approx
    2.5.
\end{equation}
Here, we used that $N_{\textrm{qubits}}(3)/ N_{\textrm{qubits}}(2)\approx 2$ and ${N_{\textrm{attempts}}(3)/N_{\textrm{attempts}}(2)} \approx 1$ at $p_{\textrm{eff},3} \approx p_{\textrm{eff},2} \approx 10^{-3}$. Recall that the $1/2$ is due to the fact that we need a $T$ gate to correct for the teleportation half of the time.

To solve the minimization problem of Eq.~\eqref{eq:synthesis_minimisation} we use an approach based on tensor network sampling akin to the recently introduced method \texttt{TRASYN}~\cite{Hao2025} for the gate sets $C_2$ and $C_3$. Details of the synthesis algorithm can be found in Appendix~\ref{app:synthesis}. Fig.~\ref{fig:synthesis} depicts the obtained trace distance $\epsilon$ as a function of total cost $R$ for the different gate sets. Fitting the data, the cost $R(\epsilon)$ scales polylogarithmically with $\epsilon$. For the set $C_3$, we find that the relative number of gates of both types in an average sequence is $\kappa = N_3(3)/N_2(2)=0.33 (0.1-1.0)$, as shown in Fig.~\ref{fig:analyse_sequences}. The average number of $T$ and $\sqrt{T}$ gates in a sequence can then be extracted as as
\begin{equation}
\begin{aligned}
    N_2(3)
    &=
    \frac{R(\epsilon)}{R_2(2) + \kappa R_3(3)},
    \\
    N_3(3) &= \kappa N_2(\epsilon),
\end{aligned}
\end{equation}
with $R_2(2)=1$ and $R_3(3)=2.5$. For $C_2$, $N_2(2) = R(\epsilon)$. Hence the number of $T$ gates used in the previous section~[Eq.~\eqref{eq:number-of-T}] matches the cost of $C_2$ in Fig.~\ref{fig:synthesis}. As a remark, we note that the total number of gates derived using the modified \texttt{TRASYN} algorithm scales with the synthesis error rate $\epsilon$ as 
\begin{equation}
    \begin{aligned}
        N_{\textrm{tot}}(2)
        =
        N_2(2) 
        &\sim 
        3.05 \log_2(1/\epsilon)
        \\
        N_{\textrm{tot}}(3)
        =
        N_2(3) + N_3(3)
        &\sim
        1.52 \log_2(1/\epsilon)
    \end{aligned}
\end{equation}
when, respectively, gates sets $C_2$ and $C_3$ are in use. For the gate set $C_k$ an asymptotic lower bound for the number of gates can be derived similar to the case for $T$-gates only. Assuming a generalization of the Matsumoto-Amano normal form~\cite{matsumoto2008representation} an information theoretic lower bound is given by
\begin{equation}
    \begin{aligned}
        N_{\textrm{tot,opt}}(k)
        &\sim 
        \frac{3}{\log_2[2(k-1)]} \log_2(1/\epsilon),
    \end{aligned}
\end{equation}
which can be derived via the same tiling argument of $SU(2)$ as in the $C_2$ case. Note that this constitutes a lower bound even if there is no straightforward generalization of the Mastumoto-Amano normal from, since additional relations among the generators of $C_k$ lead to longer asymptotic decompositions.

\begin{figure}[t]
\includegraphics[width=0.95\columnwidth]{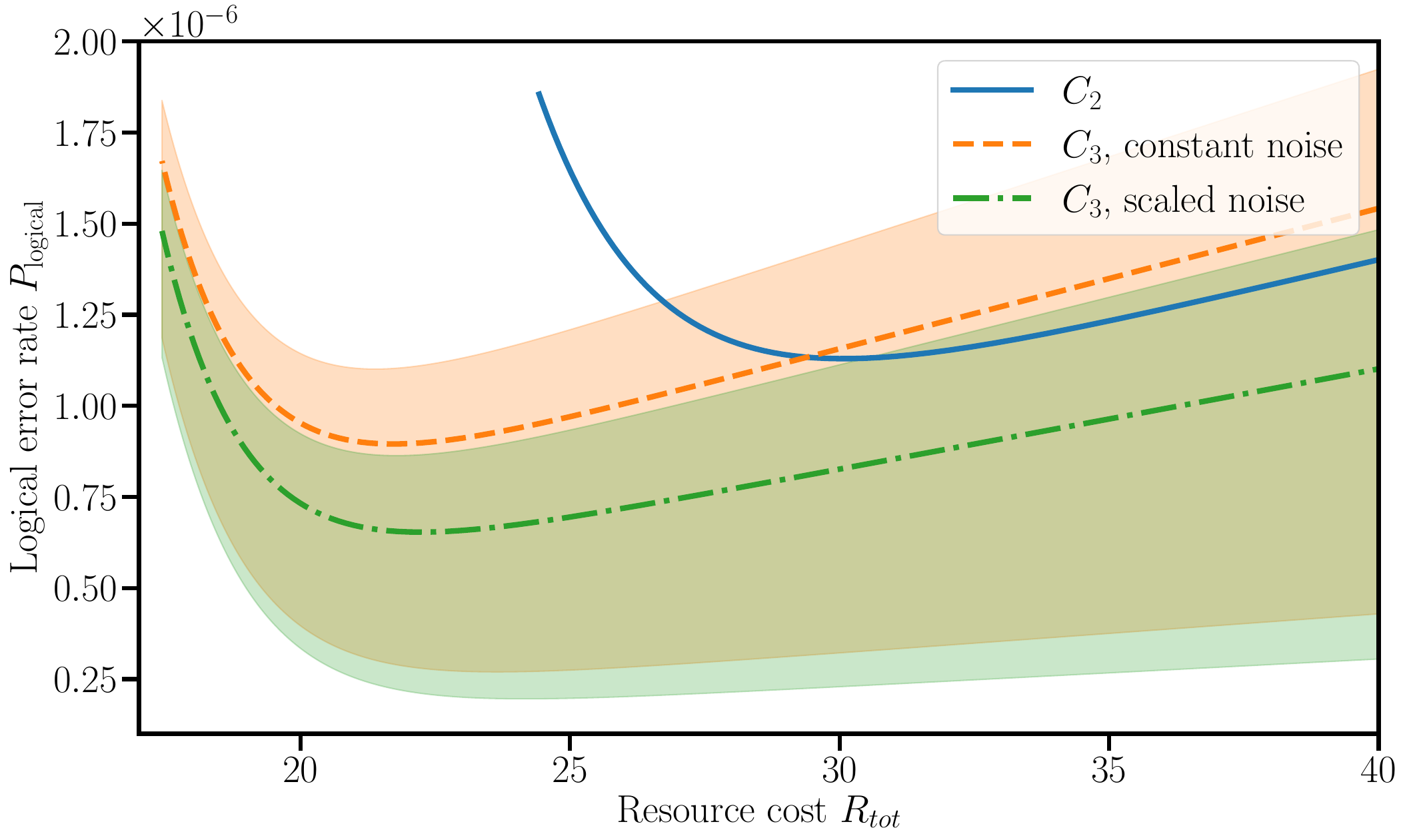}
\caption{\label{fig:ler-vs-cost} Total logical error rate~[Eq.~\eqref{eq:total-error}] versus space-time cost~[Eq.~\eqref{eq:total-cost}] of approximating an arbitrary single-qubit unitary $U$ using gate sets $C_2$ and $C_3$. The two curves shown for the set $C_3$ correspond to the constant and scaled noise models, as defined by Eqs.~\eqref{eq:noise-constant} and \eqref{eq:noise-scaled}, respectively. 
Noise parameters are the same as in Fig.~\ref{fig:ler-simulation}~(b).}
\end{figure}

The average total logical error of approximating a single-qubit unitary $U$ is a combination of decomposition and gate errors, and to the leading order reads
\begin{equation}\label{eq:total-error}
    P_{\textrm{logical}}(k)
    =
    \epsilon^2
    +
    \sum_{k' \leq k}
    N_{k'}(k)
    P_{\textrm{gate}}(k')
\end{equation}
with $P_{\textrm{gate}}(k')$ given in Eq.~\eqref{eq:cisc-seq-error}. The total logical error rate versus the total cost 
\begin{equation}\label{eq:total-cost}
    R_{\textrm{tot}}(k)
    =
    \sum_{k' \leq k}
    N_{k'}(k)
    R_{k'}(k')
\end{equation}
for $k=2$ and $k=3$ are shown in Fig.~\ref{fig:ler-vs-cost}. Both curves demonstrate a qualitatively similar behaviour. When available resource is too small, error in the total logical error explodes since the gate sequence is too short to approximate the desired unitary to a reasonable accuracy even with ideal gates. On the other hand, when long sequences of gates are used, errors accumulate from each imperfect gate and reduce the total fidelity. Hence, for each gate set $C_k$, there exist an optimal number of gates to achieve the smallest possible logical error. When comparing sets of parity-unfolded gates $C_2$ and $C_3$, we observe that the latter simultaneously reduces both the total logical noise $P_{\textrm{logical}}$ and the resource cost $R_{\textrm{tot}}$ at which it is achieved by 26\% and 49\%, respectively. Alternatively, the gate set $C_3$ can be used to achieve the smallest error achievable with $C_2$ at 61\% of the cost of the latter.

\remark{So far we have considered only the ancilla-free synthesis problem and showed that having access to higher-level non-Clifford gates such as the $\sqrt{T}$ gate is beneficial for that problem given the specific cost overhead. Current state of the art single qubit synthesis strategies also include ancilla-based methods like fallback/repeat until success protocols or unitary mixing~\cite{kliuchnikov2023shorter}. In both cases, one or multiple sub-problems have to be solved which involve approximating an entry of a single qubit unitary. For example, in the fallback protocol a projective rotation based on a unitary $V_{\textrm{FB}}$ is applied on the single qubit state and the ancilla followed by a fallback operation conditioned on the measurement outcome of the ancilla. The lower success probability~(which has to be compensated by repeating this process if needed) of the desired application of the single qubit rotation can be outweight by the simpler unitary approximation problem (since we only care about e.g. the top-left entry of $V_{\textrm{FB}}$). However, also this slightly easier problem is expected to be solved more efficiently with access to more fine-grained unitary sequences from including $\sqrt{T}$ to the non-Clifford gate set. Thus we would expect that also for these methods (and combinations thereof) our approach of using $C_{k>2}$ sets yields a better scaling of the total non-Clifford cost with the error budget.
}

\section{Conclusion and outlook}\label{sec:conclusion}

The parity-unfolded architecture introduced in this work yields significant improvement of two most relevant metrics in fault-tolerant quantum algorithm design, namely, the required space-time resource overhead and resulting error rate of the logical gate. While the most natural application of the proposed architecture is in algorithms for which the $Z_k$ gates are native, such as the quantum Fourier transform, our scheme demonstrates significant advantage in a more general task of an arbitrary unitary decomposition of $U(3)$ rotation gates. The advantage is achieved owing to two conceptual improvements introduced in this work. 

Firstly, unfolded distillation---the main building block of the unfolded architecture---offers hardware-friendly execution of quantum gates from higher levels of the Clifford hierarchy under the assumption of very strongly noise-bias platform. While gates from an $m$-th level are known to be transversal only in a code that is local in $m$-dimensional space, we provide a constructive scheme to extract any such gate in a two-dimensional qubit lattice with little to no additional qubit overhead. 

Secondly, to harness the full potential of the architecture, we equip state of the art synthesis algorithms with gates from extended sets of non-Clifford gates. Our results indicate that having access to the extended set (Clifford+$T$+$\sqrt{T}$) offers more resource-efficient execution of important quantum algorithm sub-routines. This paves the way for new research into more efficient decomposition strategies of quantum computing primitives. We note that unitary synthesis is a rapidly developing research field on its own and, while the methods developed in this work are sufficient to reach experimentally relevant logical error rates within the presented unfolding scheme, more sophisticated, scalable algorithms  are required to reach beyond the scope of his work. 

\begin{acknowledgments}
This project was supported by a FFG Funding (Project No. FO99918691) as part of the international Eureka cooperation, and by the Austrian Research Promotion Agency (FFG Project No. FO999937388, FFG Basisprogramm). For the purpose of open access, the author has applied a CC BY public copyright license to any Author Accepted Manuscript version arising from this submission.
\end{acknowledgments}

\bibliographystyle{apsrev4-1}
\bibliography{reflist}

@misc{perlin2023qldpc,
  author = {Perlin, Michael A.},
  title = {{qLDPC}},
  year = {2023},
  publisher = {GitHub},
  journal = {GitHub repository},
  howpublished = {\url{https://github.com/qLDPCOrg/qLDPC}},
}

@footnote{unfolding-simulation-repo,
	Note = {The open source code and the data used in this paper will be added in later revisions.
           }
}

@footnote{synthesis-repo,
	Note = {The open source code and the data used in this paper will be added in later revisions.
           }
}

@article{Quan_2018,
doi = {10.1088/1751-8121/aaad13},
url = {https://doi.org/10.1088/1751-8121/aaad13},
year = {2018},
month = {feb},
publisher = {IOP Publishing},
volume = {51},
number = {11},
pages = {115305},
author = {Quan, Dong-Xiao and Zhu, Li-Li and Pei, Chang-Xing and Sanders, Barry C},
title = {Fault-tolerant conversion between adjacent Reed–Muller quantum codes based on gauge fixing},
journal = {Journal of Physics A: Mathematical and Theoretical}
}

@article{matsumoto2008representation,
  title={Representation of quantum circuits with Clifford and $$\backslash$pi/8$ gates},
  author={Matsumoto, Ken and Amano, Kazuyuki},
  journal={arXiv preprint arXiv:0806.3834},
  year={2008}
}

@misc{gong2024computationquantumreedmullercodes,
      title={Computation with quantum Reed-Muller codes and their mapping onto 2D atom arrays}, 
      author={Anqi Gong and Joseph M. Renes},
      year={2024},
      eprint={2410.23263},
      archivePrefix={arXiv},
      primaryClass={quant-ph},
      url={https://arxiv.org/abs/2410.23263}, 
}

@misc{xu2026controlledjumpcliffordhierarchy,
      title={Controlled jump in the Clifford hierarchy}, 
      author={Yichen Xu and Xiao Wang},
      year={2026},
      eprint={2602.22201},
      archivePrefix={arXiv},
      primaryClass={quant-ph},
      url={https://arxiv.org/abs/2602.22201}, 
}

@article{PhysRevA.95.012329,
  title = {Diagonal gates in the Clifford hierarchy},
  author = {Cui, Shawn X. and Gottesman, Daniel and Krishna, Anirudh},
  journal = {Phys. Rev. A},
  volume = {95},
  issue = {1},
  pages = {012329},
  numpages = {7},
  year = {2017},
  month = {Jan},
  publisher = {American Physical Society},
  doi = {10.1103/PhysRevA.95.012329},
  url = {https://link.aps.org/doi/10.1103/PhysRevA.95.012329}
}

@misc{hu2021climbingdiagonalcliffordhierarchy,
      title={Climbing the Diagonal Clifford Hierarchy}, 
      author={Jingzhen Hu and Qingzhong Liang and Robert Calderbank},
      year={2021},
      eprint={2110.11923},
      archivePrefix={arXiv},
      primaryClass={quant-ph},
      url={https://arxiv.org/abs/2110.11923}, 
}

@article{Jones_2016,
   title={Gauge color codes in two dimensions},
   volume={93},
   ISSN={2469-9934},
   url={http://dx.doi.org/10.1103/PhysRevA.93.052332},
   DOI={10.1103/physreva.93.052332},
   number={5},
   journal={Physical Review A},
   publisher={American Physical Society (APS)},
   author={Jones, Cody and Brooks, Peter and Harrington, Jim},
   year={2016},
   month=may }

@article{Chamberland2020,
	author = {Chamberland, Christopher and Noh, Kyungjoo},
	date = {2020/10/27},
	doi = {10.1038/s41534-020-00319-5},
	id = {Chamberland2020},
	isbn = {2056-6387},
	journal = {npj Quantum Information},
	number = {1},
	pages = {91},
	title = {Very low overhead fault-tolerant magic state preparation using redundant ancilla encoding and flag qubits},
	url = {https://doi.org/10.1038/s41534-020-00319-5},
	volume = {6},
	year = {2020}}

@misc{xu2026distillingmagicstatesbicycle,
      title={Distilling Magic States in the Bicycle Architecture}, 
      author={Shifan Xu and Kun Liu and Patrick Rall and Zhiyang He and Yongshan Ding},
      year={2026},
      eprint={2602.20546},
      archivePrefix={arXiv},
      primaryClass={quant-ph},
      url={https://arxiv.org/abs/2602.20546}, 
}

@article{Bombín_2015,
doi = {10.1088/1367-2630/17/8/083002},
url = {https://doi.org/10.1088/1367-2630/17/8/083002},
year = {2015},
month = {aug},
publisher = {IOP Publishing},
volume = {17},
number = {8},
pages = {083002},
author = {Bombín, Héctor},
title = {Gauge color codes: optimal transversal gates and gauge fixing in topological stabilizer codes},
journal = {New Journal of Physics}
}

@article{PhysRevLett.113.080501,
  title = {Fault-Tolerant Conversion between the Steane and Reed-Muller Quantum Codes},
  author = {Anderson, Jonas T. and Duclos-Cianci, Guillaume and Poulin, David},
  journal = {Phys. Rev. Lett.},
  volume = {113},
  issue = {8},
  pages = {080501},
  numpages = {5},
  year = {2014},
  month = {Aug},
  publisher = {American Physical Society},
  doi = {10.1103/PhysRevLett.113.080501},
  url = {https://link.aps.org/doi/10.1103/PhysRevLett.113.080501}
}

@article{PhysRevResearch.7.023080,
  title = {Code switching revisited: Low-overhead magic state preparation using color codes},
  author = {Daguerre, Lucas and Kim, Isaac H.},
  journal = {Phys. Rev. Res.},
  volume = {7},
  issue = {2},
  pages = {023080},
  numpages = {12},
  year = {2025},
  month = {Apr},
  publisher = {American Physical Society},
  doi = {10.1103/PhysRevResearch.7.023080},
  url = {https://link.aps.org/doi/10.1103/PhysRevResearch.7.023080}
}

@article{PhysRevA.86.052329,
  title = {Magic-state distillation with low overhead},
  author = {Bravyi, Sergey and Haah, Jeongwan},
  journal = {Phys. Rev. A},
  volume = {86},
  issue = {5},
  pages = {052329},
  numpages = {10},
  year = {2012},
  month = {Nov},
  publisher = {American Physical Society},
  doi = {10.1103/PhysRevA.86.052329},
  url = {https://link.aps.org/doi/10.1103/PhysRevA.86.052329}
}

@inproceedings{Hao2025,
   title={Reducing T Gates with Unitary Synthesis},
   url={http://dx.doi.org/10.1145/3779212.3790210},
   DOI={10.1145/3779212.3790210},
   booktitle={Proceedings of the 31st ACM International Conference on Architectural Support for Programming Languages and Operating Systems, Volume 2},
   publisher={ACM},
   author={Hao, Tianyi and Xu, Amanda and Tannu, Swamit},
   year={2026},
   month=mar, pages={1589–1604} }

@article{Ferris2012,
  title = {Perfect sampling with unitary tensor networks},
  author = {Ferris, Andrew J. and Vidal, Guifre},
  journal = {Phys. Rev. B},
  volume = {85},
  issue = {16},
  pages = {165146},
  numpages = {10},
  year = {2012},
  month = {Apr},
  publisher = {American Physical Society},
  doi = {10.1103/PhysRevB.85.165146},
  url = {https://link.aps.org/doi/10.1103/PhysRevB.85.165146}
}

@misc{chen2009notesreedmullercodes,
      title={Notes on Reed-Muller Codes}, 
      author={Yanling Chen and Han Vinck},
      year={2009},
      eprint={0901.2062},
      archivePrefix={arXiv},
      primaryClass={cs.IT},
      url={https://arxiv.org/abs/0901.2062}, 
}

@article{10.1038/s41467-017-00045-1,
	author = {Heeres, Reinier W. and Reinhold, Philip and Ofek, Nissim and Frunzio, Luigi and Jiang, Liang and Devoret, Michel H. and Schoelkopf, Robert J.},
	date = {2017/07/21},
	doi = {10.1038/s41467-017-00045-1},
	id = {Heeres2017},
	isbn = {2041-1723},
	journal = {Nature Communications},
	number = {1},
	pages = {94},
	title = {Implementing a universal gate set on a logical qubit encoded in an oscillator},
	url = {https://doi.org/10.1038/s41467-017-00045-1},
	volume = {8},
	year = {2017}}

@misc{leroux2025romanescocodesbiastailoredqldpc,
      title={Romanesco codes: Bias-tailored qLDPC codes from fractal codes}, 
      author={Catherine Leroux and Joseph K. Iverson},
      year={2025},
      eprint={2506.00130},
      archivePrefix={arXiv},
      primaryClass={quant-ph},
      url={https://arxiv.org/abs/2506.00130}, 
}

@article{ldpc-cat-code,
	author = {Ruiz, Diego and Guillaud, J{\'e}r{\'e}mie and Leverrier, Anthony and Mirrahimi, Mazyar and Vuillot, Christophe},
	date = {2025/01/26},
	doi = {10.1038/s41467-025-56298-8},
	id = {Ruiz2025},
	isbn = {2041-1723},
	journal = {Nature Communications},
	number = {1},
	pages = {1040},
	title = {LDPC-cat codes for low-overhead quantum computing in 2D},
	url = {https://doi.org/10.1038/s41467-025-56298-8},
	volume = {16},
	year = {2025}}

@article{Lee2021surfacecode,
	author = {Lee, Jonghyun and Park, Jooyoun and Heo, Jun},
	date = {2021/07/07},
	doi = {10.1007/s11128-021-03130-z},
	id = {Lee2021},
	isbn = {1573-1332},
	journal = {Quantum Information Processing},
	number = {7},
	pages = {231},
	title = {Rectangular surface code under biased noise},
	url = {https://doi.org/10.1007/s11128-021-03130-z},
	volume = {20},
	year = {2021}}

@article{Campbell_2016,
doi = {10.1088/2058-9565/1/1/015007},
url = {https://doi.org/10.1088/2058-9565/1/1/015007},
year = {2016},
month = {dec},
publisher = {IOP Publishing},
volume = {1},
number = {1},
pages = {015007},
author = {Campbell, Earl T and O’Gorman, Joe},
title = {An efficient magic state approach to small angle rotations},
journal = {Quantum Science and Technology}
}

@article{PhysRevA.91.042315,
  title = {Reducing the quantum-computing overhead with complex gate distillation},
  author = {Duclos-Cianci, Guillaume and Poulin, David},
  journal = {Phys. Rev. A},
  volume = {91},
  issue = {4},
  pages = {042315},
  numpages = {9},
  year = {2015},
  month = {Apr},
  publisher = {American Physical Society},
  doi = {10.1103/PhysRevA.91.042315},
  url = {https://link.aps.org/doi/10.1103/PhysRevA.91.042315}
}

@footnote{footnote-circuits,
	Note = {A circuit for encoding a shortened QRM(1,4) code is given in Fig.~3 of Ref.~\cite{ruiz2025unfoldeddistillationlowcostmagic}. A circuit for encoding a shortened QRM(1,5) code is given in Fig.~3 of Ref.~\cite{landahl2013complexinstructionsetcomputing}. 
           }
}

@misc{tiurev2025optimaldecodererrorcorrecting,
      title={Optimal Decoder for the Error Correcting Parity Code}, 
      author={Konstantin Tiurev and Christophe Goeller and Leo Stenzel and Paul Schnabl and Anette Messinger and Michael Fellner and Wolfgang Lechner},
      year={2025},
      eprint={2505.05210},
      archivePrefix={arXiv},
      primaryClass={quant-ph},
      url={https://arxiv.org/abs/2505.05210}, 
}

@article{Gidney_2019,
   title={Efficient magic state factories with a catalyzed<mml:math xmlns:mml="http://www.w3.org/1998/Math/MathML"><mml:mrow class="MJX-TeXAtom-ORD"><mml:mo stretchy="false">|</mml:mo></mml:mrow><mml:mi>C</mml:mi><mml:mi>C</mml:mi><mml:mi>Z</mml:mi><mml:mo fence="false" stretchy="false">⟩</mml:mo></mml:math>to<mml:math xmlns:mml="http://www.w3.org/1998/Math/MathML"><mml:mn>2</mml:mn><mml:mrow class="MJX-TeXAtom-ORD"><mml:mo stretchy="false">|</mml:mo></mml:mrow><mml:mi>T</mml:mi><mml:mo fence="false" stretchy="false">⟩</mml:mo></mml:math>transformation},
   volume={3},
   ISSN={2521-327X},
   url={http://dx.doi.org/10.22331/q-2019-04-30-135},
   DOI={10.22331/q-2019-04-30-135},
   journal={Quantum},
   publisher={Verein zur Forderung des Open Access Publizierens in den Quantenwissenschaften},
   author={Gidney, Craig and Fowler, Austin G.},
   year={2019},
   month=apr, pages={135} }

@misc{gidney2024magicstatecultivationgrowing,
      title={Magic state cultivation: growing T states as cheap as CNOT gates}, 
      author={Craig Gidney and Noah Shutty and Cody Jones},
      year={2024},
      eprint={2409.17595},
      archivePrefix={arXiv},
      primaryClass={quant-ph},
      url={https://arxiv.org/abs/2409.17595}, 
}

@article{10.1038/s41586-024-07294-3,
	author = {R{\'e}glade, U. and Bocquet, A. and Gautier, R. and Cohen, J. and Marquet, A. and Albertinale, E. and Pankratova, N. and Hall{\'e}n, M. and Rautschke, F. and Sellem, L. -A. and Rouchon, P. and Sarlette, A. and Mirrahimi, M. and Campagne-Ibarcq, P. and Lescanne, R. and Jezouin, S. and Leghtas, Z.},
	date = {2024/05/01},
	doi = {10.1038/s41586-024-07294-3},
	id = {R{\'e}glade2024},
	isbn = {1476-4687},
	journal = {Nature},
	number = {8013},
	pages = {778--783},
	title = {Quantum control of a cat qubit with bit-flip times exceeding ten seconds},
	url = {https://doi.org/10.1038/s41586-024-07294-3},
	volume = {629},
	year = {2024}}

@article{Mirrahimi_2014,
doi = {10.1088/1367-2630/16/4/045014},
url = {https://doi.org/10.1088/1367-2630/16/4/045014},
year = {2014},
month = {apr},
publisher = {IOP Publishing},
volume = {16},
number = {4},
pages = {045014},
author = {Mirrahimi, Mazyar and Leghtas, Zaki and Albert, Victor V and Touzard, Steven and Schoelkopf, Robert J and Jiang, Liang and Devoret, Michel H},
title = {Dynamically protected cat-qubits: a new paradigm for universal quantum computation},
journal = {New Journal of Physics}
}

@article{PhysRevA.62.052316,
  title = {Methodology for quantum logic gate construction},
  author = {Zhou, Xinlan and Leung, Debbie W. and Chuang, Isaac L.},
  journal = {Phys. Rev. A},
  volume = {62},
  issue = {5},
  pages = {052316},
  numpages = {12},
  year = {2000},
  month = {Oct},
  publisher = {American Physical Society},
  doi = {10.1103/PhysRevA.62.052316},
  url = {https://link.aps.org/doi/10.1103/PhysRevA.62.052316}
}

@article{PRXQuantum.5.020345,
  title = {Fault-Tolerant Code-Switching Protocols for Near-Term Quantum Processors},
  author = {Butt, Friederike and Heu\ss{}en, Sascha and Rispler, Manuel and M\"uller, Markus},
  journal = {PRX Quantum},
  volume = {5},
  issue = {2},
  pages = {020345},
  numpages = {26},
  year = {2024},
  month = {May},
  publisher = {American Physical Society},
  doi = {10.1103/PRXQuantum.5.020345},
  url = {https://link.aps.org/doi/10.1103/PRXQuantum.5.020345}
}

@article{Litinski_2019,
   title={Magic State Distillation: Not as Costly as You Think},
   volume={3},
   ISSN={2521-327X},
   url={http://dx.doi.org/10.22331/q-2019-12-02-205},
   DOI={10.22331/q-2019-12-02-205},
   journal={Quantum},
   publisher={Verein zur Forderung des Open Access Publizierens in den Quantenwissenschaften},
   author={Litinski, Daniel},
   year={2019},
   month=dec, pages={205} }

@article{PhysRevA.95.032338,
  title = {Quantum computation with realistic magic-state factories},
  author = {O'Gorman, Joe and Campbell, Earl T.},
  journal = {Phys. Rev. A},
  volume = {95},
  issue = {3},
  pages = {032338},
  numpages = {19},
  year = {2017},
  month = {Mar},
  publisher = {American Physical Society},
  doi = {10.1103/PhysRevA.95.032338},
  url = {https://link.aps.org/doi/10.1103/PhysRevA.95.032338}
}

@article{PhysRevLett.110.170503,
  title = {Classification of Topologically Protected Gates for Local Stabilizer Codes},
  author = {Bravyi, Sergey and K\"onig, Robert},
  journal = {Phys. Rev. Lett.},
  volume = {110},
  issue = {17},
  pages = {170503},
  numpages = {5},
  year = {2013},
  month = {Apr},
  publisher = {American Physical Society},
  doi = {10.1103/PhysRevLett.110.170503},
  url = {https://link.aps.org/doi/10.1103/PhysRevLett.110.170503}
}

@misc{ruiz2025unfoldeddistillationlowcostmagic,
      title={Unfolded distillation: very low-cost magic state preparation for biased-noise qubits}, 
      author={Diego Ruiz and Jérémie Guillaud and Christophe Vuillot and Mazyar Mirrahimi},
      year={2025},
      eprint={2507.12511},
      archivePrefix={arXiv},
      primaryClass={quant-ph},
      url={https://arxiv.org/abs/2507.12511}, 
}

@misc{landahl2013complexinstructionsetcomputing,
      title={Complex instruction set computing architecture for performing accurate quantum $Z$ rotations with less magic}, 
      author={Andrew J. Landahl and Chris Cesare},
      year={2013},
      eprint={1302.3240},
      archivePrefix={arXiv},
      primaryClass={quant-ph},
      url={https://arxiv.org/abs/1302.3240}, 
}

@article{PhysRevLett.133.110601,
  title = {Domain Wall Color Code},
  author = {Tiurev, Konstantin and Pesah, Arthur and Derks, Peter-Jan H. S. and Roffe, Joschka and Eisert, Jens and Kesselring, Markus S. and Reiner, Jan-Michael},
  journal = {Phys. Rev. Lett.},
  volume = {133},
  issue = {11},
  pages = {110601},
  numpages = {6},
  year = {2024},
  month = {Sep},
  publisher = {American Physical Society},
  doi = {10.1103/PhysRevLett.133.110601},
  url = {https://link.aps.org/doi/10.1103/PhysRevLett.133.110601}
}

@article{PhysRevLett.129.180503,
  title = {Universal Parity Quantum Computing},
  author = {Fellner, Michael and Messinger, Anette and Ender, Kilian and Lechner, Wolfgang},
  journal = {Phys. Rev. Lett.},
  volume = {129},
  issue = {18},
  pages = {180503},
  numpages = {7},
  year = {2022},
  month = {Oct},
  publisher = {American Physical Society},
  doi = {10.1103/PhysRevLett.129.180503},
  url = {https://link.aps.org/doi/10.1103/PhysRevLett.129.180503}
}

@article{PhysRevA.71.022316,
  title = {Universal quantum computation with ideal Clifford gates and noisy ancillas},
  author = {Bravyi, Sergey and Kitaev, Alexei},
  journal = {Phys. Rev. A},
  volume = {71},
  issue = {2},
  pages = {022316},
  numpages = {14},
  year = {2005},
  month = {Feb},
  publisher = {American Physical Society},
  doi = {10.1103/PhysRevA.71.022316},
  url = {https://link.aps.org/doi/10.1103/PhysRevA.71.022316}
}

@article{PRXQuantum.2.020341,
  title = {Cost of Universality: A Comparative Study of the Overhead of State Distillation and Code Switching with Color Codes},
  author = {Beverland, Michael E. and Kubica, Aleksander and Svore, Krysta M.},
  journal = {PRX Quantum},
  volume = {2},
  issue = {2},
  pages = {020341},
  numpages = {46},
  year = {2021},
  month = {Jun},
  publisher = {American Physical Society},
  doi = {10.1103/PRXQuantum.2.020341},
  optUrl = {https://link.aps.org/doi/10.1103/PRXQuantum.2.020341}
}

@book{9781107002173,
  Author = {Michael A. Nielsen and Isaac L. Chuang},
  Title = {Quantum Computation and Quantum Information: 10th Anniversary Edition},
  Publisher = {Cambridge University Press},
  Year = {2011},
  ISBN = {9781107002173}
  }

@article{PhysRevX.9.041053,
  title = {Repetition Cat Qubits for Fault-Tolerant Quantum Computation},
  author = {Guillaud, J\'er\'emie and Mirrahimi, Mazyar},
  journal = {Phys. Rev. X},
  volume = {9},
  issue = {4},
  pages = {041053},
  numpages = {23},
  year = {2019},
  month = {Dec},
  publisher = {American Physical Society},
  doi = {10.1103/PhysRevX.9.041053},
  optUrl = {https://link.aps.org/doi/10.1103/PhysRevX.9.041053}
}

@article{PhysRevLett.102.110502,
  title = {Restrictions on Transversal Encoded Quantum Gate Sets},
  author = {Eastin, Bryan and Knill, Emanuel},
  journal = {Phys. Rev. Lett.},
  volume = {102},
  issue = {11},
  pages = {110502},
  numpages = {4},
  year = {2009},
  month = {Mar},
  publisher = {American Physical Society},
  doi = {10.1103/PhysRevLett.102.110502},
  optUrl = {https://link.aps.org/doi/10.1103/PhysRevLett.102.110502}
}

@article{XZZX,
title={XZZX surface code},
author={J. P. B. Ataides and D. K. Tuckett and S. D. Bartlett and S. T. Flammia  and B. J. Brown}, 
journal={Nature Comm.}, volume=12, pages={2172}, year=2021,
doi={10.1038/s41467-021-22274-1}
}

@article{mooney2021cost,
  title={Cost-optimal single-qubit gate synthesis in the Clifford hierarchy},
    url={http://dx.doi.org/10.22331/q-2021-02-15-396},
   DOI={10.22331/q-2021-02-15-396},
  author={Mooney, Gary J and Hill, Charles D and Hollenberg, Lloyd CL},
  journal={Quantum},
  volume={5},
  pages={396},
  year={2021},
  publisher={Verein zur F{\"o}rderung des Open Access Publizierens in den Quantenwissenschaften}
}

@article{kliuchnikov2023shorter,
  title={Shorter quantum circuits via single-qubit gate approximation},
  author={Kliuchnikov, Vadym and Lauter, Kristin and Minko, Romy and Paetznick, Adam and Petit, Christophe},
  journal={Quantum},
  volume={7},
  pages={1208},
  year={2023},
  publisher={Verein zur F{\"o}rderung des Open Access Publizierens in den Quantenwissenschaften}
}

\newpage
\onecolumngrid
\appendix

\newpage
\section{Proof of Lemma~\ref{lemma:1}}\label{app:lemma-1}

\textbf{Lemma 1}
Bulk stabilizers together with boundary stabilizers defined by Eq.~\eqref{eq:s-stabilizers} form an independent set of $2^m - m - 1$ stabilizers on $2^m$ physical qubits, i.e., generate a code with parameters $[n,k]$ of Eq.~\eqref{eq:rm-1-parameter}, required for the $RM(1,m)$ code.\\

\begin{proof}
First, we prove the completeness. For a fixed $w$ there are $\nq /w - 1 = 2^{l-s} - 1$ stabilizers $S(w,t)$, so the total number of stabilizers is
\begin{equation}
    \begin{aligned}
        \sum_{s=1}^{l-1}
        (2^{l-s} - 1)
        =
        2^l - l - 1, 
    \end{aligned}
\end{equation}
which matches the number of missing DOF along one side in Eq.~\eqref{eq:side-stabilizers}. 

To prove that stabilizers defined by $S(w,t)$ are independent, we first note that qubits of stabilizers come in pairs, such that each is always of the form $S(w,t) = [i,i+1, j,j+1]$. Assume $S(w,t)$ is not independent, that is, there exist stabilizers such that
\begin{equation}\label{eq:s-proof}
    S(w,t)
    =
    \prod_{q}
    S(w_q,t_q).
\end{equation}
Consider one of the pair of qubits in $S(w,t)$, e.g., the left two qubits $i$ and $i+1$. There are two possible cases of how the two qubits can belong to a product of other stabilizers. First, qubits $i$ and $i+1$ belong to two different stabilizers in the product of the RHS in Eq.~\eqref{eq:s-proof}, i.e., there are two stabilizers $S(w_1,t_1)$ and $S(w_2,t_2)$ in the product such that qubit $i$ belongs to $S(w_1,t_1)$ and qubit $i+1$ belongs to $S(w_2,t_2)$. In this case, the product in the equation above will have support on an odd number of qubits to the left of $i$ and on an odd number of qubits to the right of $i+1$. However, because qubits belonging to stabilizers always come in pairs, an odd number of qubits can not be a product of stabilizers. Hence, this configuration is not possible.

The second case is when the left pair of qubits $i$ and $i+1$ of $S(w,t)$ is also either the left of the right pair of another stabilizer from the product of Eq.~\eqref{eq:s-proof}. We will show that a right pair of stabilizer $S(w,t)$ can not be a right pair of any other stabilizer. 
Consider stabilizers $S(w,t)$ and $S(w',t')$. If the right pairs of both stabilizers match,
\begin{equation}
    \begin{aligned}
        \Big{[}        
        w(t+\frac{1}{2}),  w(t+\frac{1}{2}) + 1
        \Big{]}
        =
        \Big{[}        
        w'(t'+\frac{1}{2}),  w'(t'+\frac{1}{2}) + 1
        \Big{]},
    \end{aligned}
\end{equation}
where we subtracted $\nq/2$ on both sides. Substituting $w = 2^s$ and $w' = 2^{s'}$, we have for the first and the second index $2^{s-1}(2t+1) = 2^{s'-1}(2t'+1)$, 
which yields
\begin{equation}
    \begin{aligned}
        2^{s-s'}(2t+1) &= (2t'+1).
    \end{aligned}
\end{equation}
If $s \neq s'$, the LHS is an even number, while the RHS is an odd number, hence the equation is only valid if $s=s'$ and $t=t'$. Similarly, we can show that the right pair of stabilizer $S(w,t)$ can be the left pair of stabilizers $S(w',t')$ iff $w=w'$ and $t=t'-1$, that is, the two stabilizers are shifted versions of one another by one step. Indeed, if the right pair of $S(w,t)$ and the left pair of $S(w',t')$ match, we have
\begin{equation}
    \begin{aligned}
        \Big{[}        
        w(t-\frac{1}{2}),  w(t-\frac{1}{2}) + 1
        \Big{]}
        =
        \Big{[}        
        w'(t'+\frac{1}{2}),  w'(t'+\frac{1}{2}) + 1
        \Big{]}.
    \end{aligned}
\end{equation}
With $w = 2^s$ and $w' = 2^{s'}$, this yields
\begin{equation}
    \begin{aligned}
        2^{s-s'}(2t-1) &= (2t'+1),
    \end{aligned}
\end{equation}
which is valid iff $s=s'$ and $t - t' = 1$. Therefore, a left pair of qubits in stabilizer $S(w,k)$ can only be a right pair of stabilizers $S(w,k-1)$, meaning that all of $\nzs = 2^l - l - 1$ weight-4 stabilizers defined by Eq.~\eqref{eq:s-stabilizers} are independent. Such shifted versions of stabilizers $S(w,t)$ and $S(w,t\pm \Delta t)$ are independent by construction. The described cases cover all possible configurations of $S(w,t)$, hence, all boundary stabilizers along the horizontal boundary are independent. By symmetry, there are also $\nzs$ boundary stabilizers along the vertical boundary. 

We have shown that the number of independent boundary stabilizers along the boundary of each type is $\nzs$. The total number of independent weight-4 stabilizers is hence $\nzb + 2\nzs = \nz$, which forms the full generating set of $Z$ stabilizers required for the distillation of $Z_{m-2}$ as shown in Fig.~\ref{fig:distillation}. That is, fixing $\nzb$ independent bulk stabilizers and $2\nzs$ independent boundary stabilizers $S(w,t)$ on $2^m$ physical qubits defines the unfolded code $uRM(m)$ with parameters $[n,k]$ of Eq.~\ref{eq:rm-1-parameter}. 
\end{proof}

\newpage
\section{Proof of Lemma~\ref{lemma:2}} \label{app:lemma-2}

\textbf{Lemma 2}
The parity labels of qubits in $uRM(m)$ code are all possible odd-length combinations of $m+1$ logical indices.\\

\begin{proof}

Consider an unfolded code $uRM(m)$. The code encodes $m+1$ logical qubits into $2^m$ physical qubits. The number of odd-length combinations one can make from integer indices $q \in [0,m]$ is
\begin{equation}\label{eq:number-of-odd-parities}
    \sum_{i \in \textrm{odd}}^m \binom{m+1}{i}
    =
    2^m,
\end{equation}
which is exactly the number of qubits.

Consider an explicit parity layout construction, such as the one of Fig.~\ref{fig:bulk}. The layout encodes $2l+1$ logical qubits and is defined on a squared layout with $2^{l}\times 2^{l}$ physical qubits. First, consider the bottom row of the code. There are $2^l$ physical qubits and $\nzs = 2^l - l - 1$ boundary stabilizers. Therefore, we can place $2^l - \nzs = l+1$ base~(i.e., single-label) qubits along the bottom row. The remaining $2^l - l - 1$ physical qubits along the bottom row have multi-label parity indices. We will show that no even-length parity indices are present along the bottom row by induction. 

\textbf{Base case.} 
Consider a weight-4 boundary stabilizers that has a support on three base qubits, i.e., on parity qubits with single-label parity indices. Then, the fourth parity qubit has an odd-length parity index. 

This is trivial to show. Each parity index of qubits belonging to a stabilizer should enter an even number of times. Then, the parity index of the fourth qubit is the symmetric different between parity indices of the other three. The latter are there different single-label indices $i$, $j$, and $k$. Hence, the fourth qubit has an odd-length~(length-3) parity index $ijk$.

\textbf{Induction.}
Assume a weight-4 boundary stabilizer with three qubits having odd-length parity indices. Then, the parity index of the fourth qubit has odd length. 

Denote qubits in the support of a stabilizer 1 to 4. $C_q$ is an parity index of qubit $q \in [1,4]$, i.e., a some combination of indices $\{1,2,...,m+1\}$. Qubits 1,2,3 have having odd-length parity indices, $|C_q|\bmod 2 = 1$ for $q=\{1,2,3\}$. Because each parity index around the stabilizer must enter an even number of times, we have $C_4 = C_1 \Delta C_2 \Delta C_3$ and, using associativity, 
\begin{equation}
\begin{aligned}
    |C_4| \bmod 2
    &=
    |C_1 \Delta C_2 \Delta C_3| \bmod 2
    \\&=
    (|C_1| + |C_2| + |C_3|)
    \bmod 2
    =
    1.
\end{aligned}
\end{equation}
Hence, a weight-4 stabilizer generates an odd-length parity label from three odd-length parity labels. Starting from single-qubit labels~(Base), we generate all possible odd-length parity labels. Because all $\nzs$ boundary stabilizers are independent, all parity labels along the bottom boundary have to be unique.

In the same manner, we generate parity labels of qubits along the rightmost column. Finally, we generate all the remaining parity labels of the code step by step from weight-4 stabilizers with support on three previously generated odd-length parity labels. By the induction, all parity labels generated this way are odd-length. Because all $\nzb$ bulk stabilizers are independent, all parity labels along have to be unique. Since all parity labels are unique and odd length, and due to Eq.~\eqref{eq:number-of-odd-parities}, this process generates all possible odd-length parity labels. A step-by-step process of assigning parity labels is described in Fig.~\ref{fig:bulk}.
\end{proof}

\newpage
\section{Proof of Lemma~\ref{lemma:3}} \label{app:lemma-3}

\textbf{Lemma 3}
Distance of each logical operator in the $uRM(m)$ code is $2^{m-1}$.\\

\begin{proof}

A logical operator $i$ of the $uRM(m)$ code has support on all physical qubits with parity index containing single-qubit index $i$. Hence we need to show that the size of any subset of parity indices belongs to $2^{m-1}$ qubits of the $uRM(m)$ code. 

Let $S$ be a set of all parity labels. If an element $s$ of $S$ contains single-qubit label $i$, the remaining part of the parity label must have even length $k$ by Lemma~\ref{lemma:2}. The number of ways to choose $k$ elements from the remaining $m$ integers is $\binom{m}{k}$, and the number of all elements $s \in S$ that obey this property is
\begin{equation}
    \sum_{k=even}
    \binom{m}{k}
    =
    2^{m-1}.
\end{equation}
This applies to each single-qubit label $i \in [0,m]$.
\end{proof}

\newpage
\section{Transversality of $Z_k$ gates from the parity formalism}\label{app:transversality-criteria}

Here we provide an alternative proof that the unfolded parity code $uRM(m)$ supports the gate $Z_{m-2}$ transversally. We also introduce a criteria we call $k$-parity that might be applicable to analysing transversal gates in a broader range of error correcting codes. 

Consider a classical parity code that generates a set of parity labels $S$. We will call the code $k$-parity if each length-$k$ combination of single-qubit parity labels enters $S$ an even number of times. Then,
\begin{lemma}\label{lemma:5}
The $uRM(m)$ code is a $(m-1)$-parity code.
\end{lemma}
\begin{proof}
     Consider the $uRM(m)$ code. It encodes $(m+1)$ logical qubits, i.e., has $(m+1)$ single-qubit labels. Pick a parity label made of $m$ single-qubit parity labels. If $m$ is even, such element enters only two elements--itself and the length-$(m+1)$ element made of all single-qubit parity labels. If $m$ is odd, the largest length of parity labels is $m$ single-qubit parity labels, and there are only two such parity labels, where $(m-1)$ labels are fixed the the remaining one can be one of two unused labels. Hence, for any choice of $m$, any parity label containing the combination of $m$ single-qubit labels enters exactly twice. The code is therefore $(m-1)$-parity.
\end{proof}

We now revisit the well-known $k$-orthogonality criteria of quantum error correcting codes. Consider an arbitrary quantum error correcting code $\mathcal{C}$. Let $S_x$ be a $r \times n$ parity check matrix. We say that $S_x$ satisfies the $k$-orthogonality condition if, for every subset of $k$ rows $\{r_{i_1}, r_{i_2}, \dots, r_{i_k}\}$ from $S_x$, the following condition holds,
\begin{equation}\label{eq:k-orthogonality}
    \sum_{j=1}^{n} (r_{i_1})_j
    \cdot
    (r_{i_2})_j
    \dots 
    (r_{i_k})_j 
    \equiv 
    0 \pmod 2. 
\end{equation}
Another words, the intersection of any $k$ rows of the matrix must have an even number of elements. Such codes are referred to as $k$-orthogonal. Any $(k+1)$-orthogonal code is known to support transversal implementation of the $Z_{k}$ gate. For instance, tri-orthogonal codes are commonly used to determine codes with orthogonal $T$ gate. 

We now show the connection between $k$-orthogonality and $k$-parity. Promote all $X$ stabilizers of the code $\mathcal{QC}$ to logical operators. Then each row of $S_x$ describes a logical operator of a classical code $\mathcal{CC}$, i.e., a parity code. The following Lemma then holds.
\begin{lemma}\label{lemma:4}
    $\mathcal{QC}$ is $k$-orthogonal if and only if $\mathcal{CC}$ is $k$-parity.
\end{lemma}
\begin{proof}
    In the parity code formalism, each qubit acquires a parity label, i.e., each column correspond to a parity index. From $k$-orthogonality of $S_x$, the intersection of any $k$ rows has an even number of elements, meaning that any $k$ logical operators of $\mathcal{CC}$ has even number of common parity labels containing all $k$ single-qubit labels. Hence, $k$-parity is simply a reformulation of $k$-orthogonality using the parity code language.
\end{proof}

From Lemma~\ref{lemma:4} and Lemma~\ref{lemma:5}, the $uRM(m)$ code is $(m-1)$ orthogonal, hence, admits transversal gate $Z_{m-2}$.

\newpage
\section{Unfolded layouts with nearest-neighbour connectivity}\label{app:nn-layout}

Figures~\ref{fig:boundaries-16} and \ref{fig:boundaries-32} provide explicit construction of boundary stabilizers realized with nearest-neighbour connectivity up to the code $uRM(10)$, i.e., for magic factories for distillation of gates up to $Z_8=T^{1/64}$. Higher-level factories can be constructed using a similar construction. We stop at $k=8$ since, as we show in Sec.~\ref{subsec:cost-native}, this is the the break-even beyond which distilling higher-order gates becomes more resource-demanding than approximation with (Clifford+$T$), at least for use-cases and parameters considered in this paper. 

\begin{figure*}[h]
\includegraphics[width=0.5\textwidth]{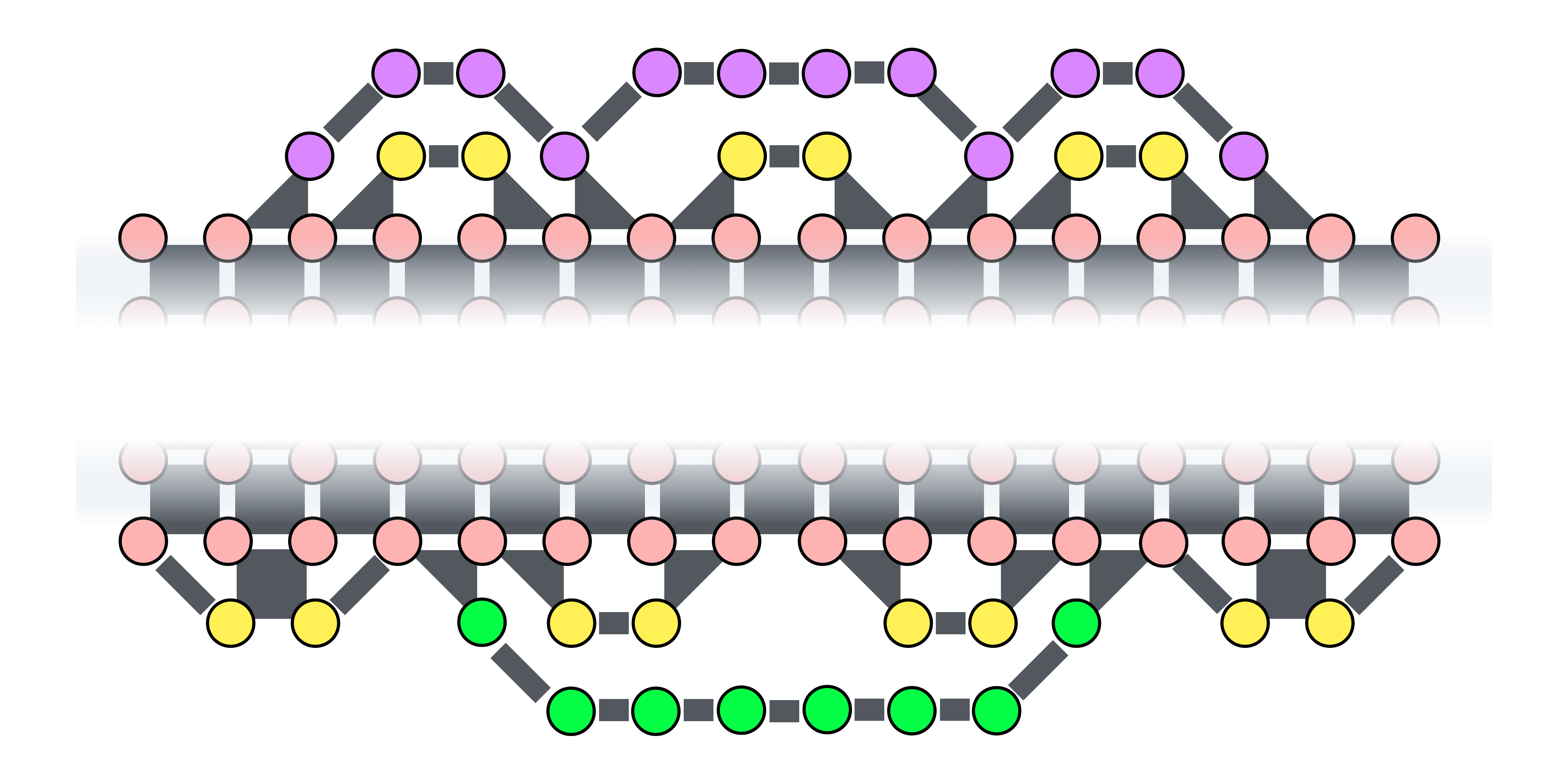}
\caption{\label{fig:boundaries-16} Boundary stabilizers with nearest-neighbour qubit connectivity. Shown are stabilizers along the horizontal boundaries of the $uRM(m)$ codes with a side length $2^4$, i.e., $\overline{QRM}(1,7)$ and $\overline{QRM}(1,8)$. Only top and bottom rows of the code are shown to save space. Stabilizers of different distances are shown with different colors. Distance-$w$ stabilizers requires exactly $w$ ancilla qubits, making it the lowest-weight possible realization on a planar chip with local interaction.}
\end{figure*}

\begin{figure*}[h]
\includegraphics[width=0.75\textwidth]{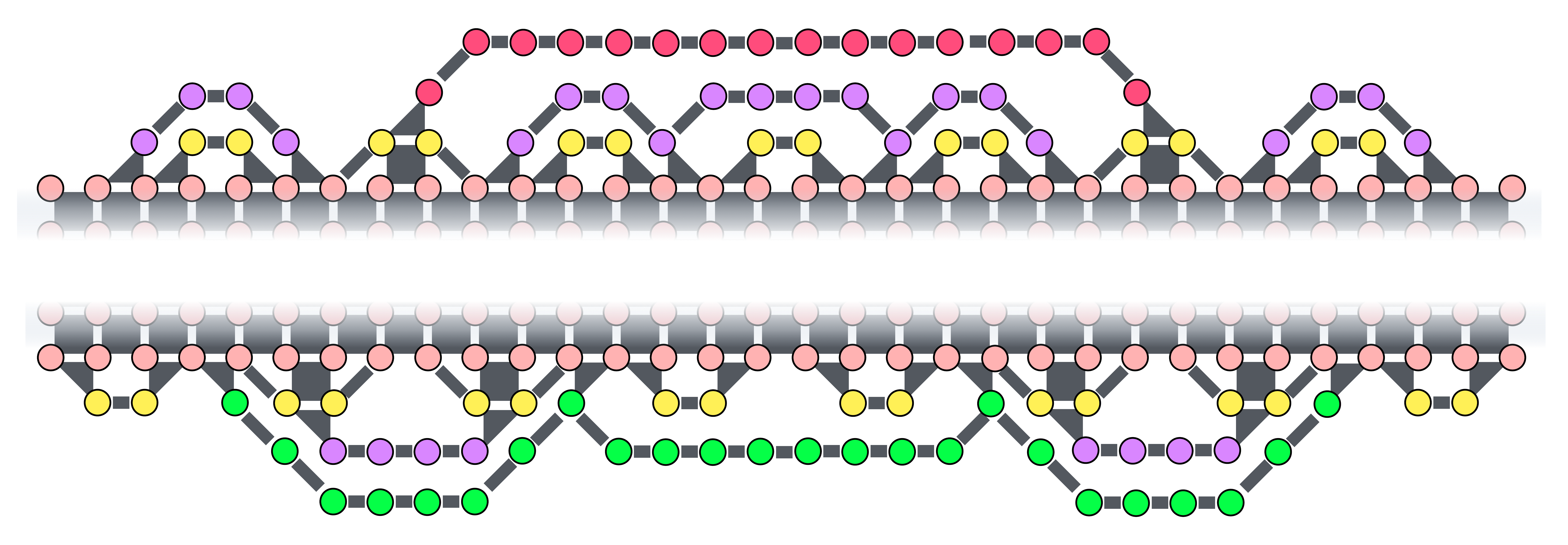}
\caption{\label{fig:boundaries-32} Same as in Fig.~\ref{fig:boundaries-16}, but for codes with a side length $2^4$, i.e., $\overline{QRM}(1,9)$ and $\overline{QRM}(1,10)$.}
\end{figure*}

\newpage
\section{Unitary synthesis}\label{app:synthesis}
To synthesise arbitrary single-body unitaries using the gate sets $C_2$ and $C_3$ we use a modified version of the recently introduced method \texttt{TRASYN} \cite{Hao2025}. \texttt{TRASYN} uses a tensor network approach to find suitable unitaries $V$ which minimize the distance to the target unitary $U$ as defined in Eq.~\eqref{eq:synthesis_error}. This requires three steps:  First, for a given set of gates $C$, all unique unitaries formed from combinations of gates $\in C$ up to a given non-Clifford cost $R$ are pre-computed (Clifford gates are assumed free). For each cost $c$, this yields a sequence of $M$ $2 \times 2$ matrices which we reshape into a single tensor $\tau_R$ with shape $2\times M\times 2$. Next, to find a minimal synthesis error for a given cost $R$, we compute the contraction of this tensor with the target unitary and sample from the resulting vector of length $M$. Since the number $M$ of possible matrices $V$ scales exponentially with the total cost, this brute force approach is very limited. To reach beyond the brute force scale, \texttt{TRASYN} introduces a tensor network $ \tau_\ell = [\tau_{R_1}, \tau_{R_2} , ..., \tau_{R_\ell}]$ formed from $\ell$ tensors $\tau_{R_i}$ with shape $2\times M_i\times 2$ and bond-dimension $2$. Each tensor $\tau_{R_i}$ contains all the unique matrices associated with a fixed non-Clifford cost $R_i$. Thus, any combination of matrices sampled from the tensors $\tau_{R_i}$ has a total cost $R= \sum_{i=1}^\ell R_i$. 

To evaluate the synthesis errors we contract $\tau_\ell$ with the target unitary $U$. \texttt{TRASYN} then adopts established sampling techniques \cite{Ferris2012} to find gate sequences with low synthesis errors. Concretely, to sample efficiently, the tensor network $\tau_\ell$ is contracted with $U$ and brought to canonical form using repeated singular value decomposition so that eventually only one tensor is non-isometric. All other tensors contract to the identity with their conjugate transpose. This allows to interpret the trace values as a joint probability distribution over all indices $M_1, M_2, \dots, M_\ell$ which can be reformulated in terms of conditional probabilities
\begin{equation}
    p(M_1, M_2, \dots, M_\ell) =  p(M_1) p(M_2 | m_1) \dots p(M_\ell | m_1, m_2, ..., m_{\ell-1}).
\end{equation}
The canonical form allows to efficiently compute the conditional probability distributions iteratively and sample in each step. This yields a list of indices $[m_1, m_2, ...,m_\ell]$ from which we can reconstruct a unitary by selecting the corresponding $2 \times 2$ matrices from the respective tensors $\tau_{R_i}$
\begin{equation}
    V_{i,j} = \tau_{R_1}^{i, m_1, k} \tau_{R_2}^ {k, m_2, n} \dots\tau_{R_{\ell-1}}^{p, m_{\ell-1}, o} \tau_{R_\ell}^{o, m_{\ell}, j},
\end{equation}
where the superscript indicates indices of the corresponding tensor using Einstein summation convention.

In Ref.~\cite{Hao2025} the authors use this method to synthesise arbitrary single-qubit unitaries from the gate set $C_2$ containing the gates $\{T,~S,~H,~X,~Y,~Z\}$ for a given non-Clifford budget. In this work we extend the gate sets to $C_3=\{T^{1/2},T,~S,~H,~X,~Y,~Z\}$. 
We proceed along similar lines: first, we pre-compute all unique unitaries up to a non-Clifford cost $R_{\mathrm{trunc}}$. The extended gate set leads to a rapid exponential growth in the number of unique matrices which limits $R_{\mathrm{trunc}}$ to moderate values ($R_{\mathrm{trunc}}=8$ for the results presented in this work). To synthesise unitaries beyond the cost $R_{\mathrm{trunc}}$, we construct tensor networks with a total cost $R>R_{\mathrm{trunc}}$ from fixed cost tensors $\tau_{R_i}$, where now $\tau_{R_i}$ contains all unique matrices formed from combinations of gates in the respective gate set. Then, we use the sampling approach outlined above to extract a unitary $V$ with a small synthesis error and total non-Clifford cost $\leq R$. 

As outlined in the main text, we attribute a non-Clifford cost of $1$ to any $T$ while $T^{1/2}$ 
contributes a cost of $2.5$. 
The cost distribution of the different non-Clifford gates involved complicates the partitioning of a total cost $R>R_{\mathrm{trunc}}$ into multiple smaller costs $R_i$. If there is only one type of non-Clifford gate and a total cost $R>R_{\mathrm{trunc}}$ we just need to find a partition of $R$ into multiple smaller costs $R_i$ such that $\sum_i R_i = R$ and $R_i \leq R_{\mathrm{trunc}} \forall i$. However, if there are different types of non-Clifford gates with different costs, a simple partitioning might not be sufficient. Consider the gate set $C_3$ and assume $R_{\mathrm{trunc}}=4$ and $R=5$. We could suggest a simple partitioning of $R$ into $R_1=4$ and $R_2=1$ build the corresponding tensor and sample from it according the scheme outlined above. However, neither of the two constituent tensors $\tau_{R_1}$ and $\tau_{R_2}$ contains any $T^{1/2}$ (since the cost of $T^{1/2}$ is $2.5$  which does neither fit in $R_1$ nor $R_2$). Conversely, we could have chosen a partitioning of $R$ into $R_1=2.5$ and $R_2=2.5$ which restricts us to tensors that only contain $T^{1/2}$. Since no single partitioning is sufficient for the case of multiple non-Clifford gate types, we need to perform the tensor network sampling on a set of partitionings $\mathcal{P}$. $\mathcal{P}$ has to be constructed such that any combination and order of the respective non-Clifford gates at fixed cost is at least contained in one partitioning $p \in \mathcal{P}$. For the example above, one such $\mathcal{P}$ is given by ${(4,1), (2.5,2.5)}$. Then we perform the tensor network sampling algorithm for each partitioning in $\mathcal{P}$ and return the best result obtained. Our code for gate synthesis is available under~\cite{synthesis-repo}.

\twocolumngrid

\end{document}